\documentclass[sigconf, authorversion]{acmart}

\AtBeginDocument{%
  \providecommand\BibTeX{{%
    \normalfont B\kern-0.5em{\scshape i\kern-0.25em b}\kern-0.8em\TeX}}}

\copyrightyear{2023}
\acmYear{2023}
\setcopyright{acmlicensed}\acmConference[KDD '23]{Proceedings of the 29th ACM SIGKDD Conference on Knowledge Discovery and Data Mining}{August 6--10, 2023}{Long Beach, CA, USA}
\acmBooktitle{Proceedings of the 29th ACM SIGKDD Conference on Knowledge Discovery and Data Mining (KDD '23), August 6--10, 2023, Long Beach, CA, USA}
\acmPrice{15.00}
\acmDOI{10.1145/3580305.3599426}
\acmISBN{979-8-4007-0103-0/23/08}

\usepackage{amsmath,amsfonts,bm}

\def\eqref#1{equation~\ref{#1}}

\def\1{\bm{1}}

\DeclareMathAlphabet{\mathsfit}{\encodingdefault}{\sfdefault}{m}{sl}
\SetMathAlphabet{\mathsfit}{bold}{\encodingdefault}{\sfdefault}{bx}{n}

\usepackage{hyperref}
\usepackage{url}
\usepackage{bibunits}
\usepackage{xspace}
\usepackage{amsthm}
\usepackage{graphicx}
\usepackage{subfigure}
\usepackage{wrapfig}
\usepackage{threeparttable}
\usepackage{xcolor}
\usepackage{caption}
\usepackage{tabularx}
\usepackage{booktabs,multirow,diagbox,enumitem}
\usepackage{float}
\usepackage{algorithm,algorithmic}
\usepackage{multicol}

\newcommand{\model}{\textit{MBrain}\xspace}

\newtheorem{task}{Task}
\newtheorem{theorem}{Theorem}
\newtheorem{proposition}{Proposition}

\newcommand{\vpara}[1]{\vspace{0.05in}\noindent\textbf{#1 }}

\begin{document}

\title{MBrain: A Multi-channel Self-Supervised Learning Framework for Brain Signals}

\author{Donghong Cai}
\authornote{Both authors contributed equally to this research.}
\affiliation{%
  \institution{Zhejiang University}
  \country{}
  }
\email{donghongcai@zju.edu.cn}
\orcid{0009-0005-5790-7505}

\author{Junru Chen}
\authornotemark[1]
\affiliation{%
  \institution{Zhejiang University}
  \country{}
  }
\email{jrchen_cali@zju.edu.cn}
\orcid{0000-0002-5989-2897}

\author{Yang Yang}
\authornote{Corresponding author.}
\affiliation{%
  \institution{Zhejiang University}
  \country{}
  }
\email{yangya@zju.edu.cn}
\orcid{0000-0002-5058-4417}

\author{Teng Liu}
\affiliation{%
  \institution{Zhejiang University}
  \country{}
  }
\email{liuteng_27@zju.edu.cn}
\orcid{0009-0007-7040-1378}

\author{Yafeng Li}
\affiliation{%
  \institution{Nuozhu Technology Co., Ltd.}
  \country{}
  }
\email{yafeng.li@neurox.cn}
\orcid{0009-0001-9681-939X}

\begin{abstract}

Brain signals are important quantitative data for understanding physiological activities and diseases of human brain. 
Meanwhile, rapidly developing deep learning methods offer a wide range of opportunities for better modeling brain signals, which has attracted considerable research efforts recently. 
Most existing studies pay attention to supervised learning methods, which, however, require high-cost clinical labels. 
In addition, the huge difference in the clinical patterns of brain signals measured by invasive (\emph{e.g.}, SEEG) and non-invasive (\emph{e.g.}, EEG) methods leads to the lack of a unified method. 
To handle the above issues, in this paper, we propose to study the self-supervised learning (SSL) framework for brain signals that can be applied to pre-train either SEEG or EEG data. 
Intuitively, brain signals, generated by the firing of neurons, are transmitted among different connecting structures in human brain. 
Inspired by this, we propose \model to learn implicit spatial and temporal correlations between different channels (\emph{i.e.}, contacts of the electrode, corresponding to different brain areas) as the cornerstone for uniformly modeling different types of brain signals. 
Specifically, we represent the spatial correlation by a \textit{graph} structure, which is built with proposed multi-channel CPC.
We theoretically prove that optimizing the goal of multi-channel CPC can lead to a better predictive representation and apply the \textit{instantaneou-time-shift prediction} task based on it. 
Then we capture the temporal correlation by designing the \textit{delayed-time-shift prediction} task.
Finally, \textit{replace-discriminative-learning} task is proposed to preserve the characteristics of each channel.
Extensive experiments of seizure detection on both EEG and SEEG large-scale real-world datasets demonstrate that our model outperforms several state-of-the-art time series SSL and unsupervised models, and has the ability to be deployed to clinical practice.

\end{abstract}



\maketitle

\section{Introduction}

Brain signals are foundational quantitative data for the study of human brain in the field of neuroscience. 
The patterns of brain signals can greatly help us to understand the normal physiological function of the brain and the mechanism of related diseases. 
There are many applications of brain signals, such as cognitive research~\citep{Ismail2020eegCognitive, Shiba2018eegCognitive}, emotion recognition~\citep{song2020eegEmotion, chen2019eegEmotion}, neurological disorders~\citep{Alturki2020eegDisorder, Ye2019eegSeizure} and so on.
Brain signals can be measured by noninvasive or invasive methods~\citep{PALUSZEK2015297}. 
The noninvasive methods, like \textit{electroencephalography} (EEG), cannot simultaneously consider temporal and spatial resolution along with the deep brain information, but they are easier to implement without any surgery. 
As for invasive methods like \textit{stereoelectroencephalography} (SEEG), they require extra surgeries to insert the recording devices, but have access to more precise and higher signal-to-noise data.
For both EEG and SEEG data, there are multiple \textit{electrodes} with several
contacts (also called \textit{channels}) that are sampled at a fixed frequency to record brain signals.

Recently, discoveries in the field of neuroscience have inspired advances of deep learning techniques, which in turn promotes neuroscience research. 
According to the literature, most deep learning-based studies 
of brain signals 
focus on supervised learning~\citep{shoeibi2021epileptic, rasheed2020machine, zhang2021survey, craik2019deep}, which relies on a large number of clinical labels. However, 
obtaining accurate and reliable clinical labels requires a high cost. 
In the meantime, the emergence of self-supervised learning (SSL) and its great success~\citep{chen2021simsiam, Brown2020GPT3, devlin2018bert, van2018representation} makes it a predominant learning paradigm in the absence of labels.
Therefore, some recent studies have introduced the means of SSL to extract the representations of brain signal data. 
For example, \citet{Banville2021SSLEEG} directly applies general SSL tasks to pre-train EEG data, including relative position prediction~\citep{doersch2015unsupervised}, temporal shuffling~\citep{misra2016shuffle} and contrastive predictive coding~\citep{van2018representation}. \citet{Mohsenvand2020CRLforEEG} designs data augmentation methods, and extends the self-supervised model SimCLR~\citep{chen2020simple} in 
computer vision to EEG data. 
In contrast to numerous works investigating EEG, few studies focus on SEEG data. 
\citet{martini2021deep} proposes an SSL model for real-time epilepsy monitoring in multimodal scenarios with SEEG data and video recordings. 

Despite the advances on representation learning of brain signals, two main issues remain to be overcome. 
\textbf{Firstly, almost all existing methods are designed for a particular type of brain signal data, and there is a lack of a unified method for handling both EEG and SEEG data.}
The challenge mainly lies in the different clinical patterns of brain signals that need to be measured in different ways. 
On the one hand, EEG collects noisy and rough brain signals on the scalp;  
differently, SEEG collects deeper signals with more stereo spatial information, which indicates more significant differences of different brain areas~\citep{perucca2014intracranial}.
On the other hand, in contrast to EEG with a gold-standard collection location, the monitoring areas of SEEG vary greatly between subjects, leading to different number and position of channels. 
Therefore, how to find the commonalities of EEG and SEEG data to design a unified framework is challenging.

\textbf{Another issue is the gap between existing methods and the real-world applications.} In clinical scenarios, doctors typically locate brain lesions by analyzing signal patterns of \textit{each} channel and their \textit{holistic} correlations. 
A straight-forward way for this goal is to model each of the channels separately by single-channel time series models, which, however, 
cannot exploit correlations between brain areas~\citep{davis2020spontaneous, lynn2019physics}.
As for the existing multivariable time series models, most of them can
only capture implicit correlation patterns~\citep{zerveas2021transformer,chen2021multi}, whereas explicit correlations are required by doctors for identifying lesions.
Moreover, although some graph-based methods have been proposed to explicitly learn correlations, they focus on giving an overall prediction for all channels at a time but overlook the prediction on one specific channel~\citep{zhang2021graph, shang2021discrete}.
Therefore, how to explicitly capture the spatial and temporal correlations 
while giving channel-wise prediction is another issue to be overcome. 

To address the challenges above, we propose a 
multi-channel self-supervised learning framework \model, which can be generally applied for learning representations of both EEG and SEEG data. 
Specifically, based on domain knowledge and data observations, we propose to learn the correlation graph between channels as the common cornerstone for both two types of brain signals. 
In particular, we employ Contrastive Predictive Coding (CPC)~\citep{van2018representation} as the backbone model of our framework by extending it to handle multi-channel data.
We theoretically prove that the optimization objective of the proposed multi-channel CPC is to maximize the mutual information of each channel and its correlated ones, so as to obtain better predictive representations.
Based on the multi-channel CPC, 
we propose the instantaneous time shift task to explicitly learn the spatial correlations between channels,
and the delayed time shift task and the replace discriminative task are designed
to capture the temporal correlation patterns and to preserve the characteristics of each channel respectively.
To validate the effectiveness of our model, we pay special attention to its application in seizure detection.
Extensive experiments show that \model outperforms several state-of-the-art baselines on large-scale real-world EEG and SEEG datasets for the seizure detection task.
Overall, the main contributions of this work can be summarized as follows:

\begin{itemize}[leftmargin=*]
    \item We are the first work to design a generalized self-supervised learning framework, which can be applied to pre-train both EEG and SEEG signals.
    
    \item We propose \model to explicitly capture the spatial and temporal correlations of brain signals to learn a unique representation for each channel.
    
    \item We validate the effectiveness and clinical value of the proposed framework through extensive experiments on large-scale real-world EEG and SEEG datasets.
\end{itemize}

\section{Preliminary: Theoretical Analysis of Multi-channel CPC }\label{sec:theory}

We employ Contrastive Predictive Coding (CPC)~\citep{van2018representation} as the basis of our framework. 
The pretext task of CPC is to predict low-level local representations by high-level global contextual representations $c_{t}$ at the $t$-th time step. Theoretically, the optimal InfoNCE loss proposed by CPC with $N-1$ negative samples $\mathcal{L}_{N}^{\text{opt}}$ is a lower bound of the mutual information between contextual semantic distribution $p(c_{t})$ and raw data distribution $p(x_{t+k})$, \emph{i.e.}, $\mathcal{L}_{N}^{\text{opt}} \ge -I(x_{t+k};c_{t}) + \log{N}$, where $k$ is the prediction step size.
CPC is originally designed for single-channel sequence data only, and 
there are two natural ways to extend single channel CPC to multi-channel version. The first one is to use CNNs with multiple kernels to encode all channels simultaneously, which cannot offer explicit correlation patterns for doctors to identify lesions. The second one is to train a shared CPC regarding all channels as one, which has no ability to capture the correlation patterns.
Taking a comprehensive consideration, we propose 
\textit{multi-channel CPC} in this paper. 
Our motivation is to explicitly aggregate the semantic information of multiple channels to predict the local representations of one channel.
Formally, we propose the following proposition as our basic starting point.
\begin{proposition}
    Introducing the contextual information of the correlated channels increases the amount of mutual information with the raw data of the target channel. 
    \begin{equation}
        I(x_{t+k}^{i};\Phi(c_{t})) = I(x_{t+k}^{i};c_{t}^{i}, \Phi(\{c_{t}^{j}\}_{j \neq i})) \ge I(x_{t+k}^{i}; c_{t}^{i}),
    \end{equation}
    where $i$ and $j$ are indexes of the channels. $\Phi(\cdot)$ represents some kinds of aggregate function, which has no additional formal constraints other than the need to retain information of the target channel.
\end{proposition}
\begin{proof}
    We use the linear operation of mutual information to obtain: $I(x_{t+k}^{i};c_{t}^{i}, \Phi(\{c_{t}^{j}\}_{j \neq i})) = I(x_{t+k}^{i}; c_{t}^{i}) + I(x_{t+k}^{i};\Phi(\{c_{t}^{j}\}_{j \neq i}) | c_{t}^{i})$. 
    According to the non-negativity of the conditional mutual information, we complete the proof.
\end{proof}
It seems natural that the predictive ability of multiple channels is stronger than that of a single channel, which is also consistent with the assumption of Granger causality~\citep{granger1969investigating} to some extent. Therefore, we choose to approximate the more informative $I(x_{t+k}^{i};\Phi(c_{t}))$ to obtain more expressive representations. Specifically, followed by InfoNCE, we define our loss function $\mathcal{L}_{N}$ as 
\begin{equation}
    \mathcal{L}_{N} = -\sum_{i}\mathbb{E}_{X^{i}}\left[ \log{\frac{f_{k}(x_{t+k}, \Phi(c_{t}))}{\sum_{x_{j} \in X} f_{k}(x_{j}, \Phi(c_{t}))}} \right], \label{eq:multi_cpc}
\end{equation}
where $X^{i}$ denotes the data sample set consisting of one positive sample and $N-1$ negative samples of the $i$-th channel.
We then establish the relationship between $\mathcal{L}_{N}$ and $I(x_{t+k}^{i};\Phi(c_{t}))$.
\begin{theorem}
    Given a sample set for each channel $X^{i}=\{ x_{1}^{i}, \dots, x_{N}^{i} \}$, $i=1,\dots,n$ consisting of one positive sample from $p(x_{t+k}^{i}|\Phi(c_{t}))$ and $N-1$ negative samples from $\sum_{j}p(x_{t+k}^{j})/n$, where $n$ is the number of channels. The optimal $\mathcal{L}_{N}^{\text{opt}}$ is the lower bound of $\sum_{i} I(x_{t+k}^{i};\Phi(c_{t}))$:
    \begin{equation}
        \mathcal{L}_{N}^{\text{opt}} \ge \sum_{i} \left[ -I(x_{t+k}^{i};\Phi(c_{t})) + \log{N} \right].
    \end{equation}
\end{theorem}
\begin{proof}
    The optimal $f_{k}(x_{t+k}, \Phi(c_{t}))$ is proportional to the division of two distributions $p(x_{t+k}^{i}|\Phi(c_{t}))/(\sum_{j}p(x_{t+k}^{j})/n)$, which is the same as single-channel CPC.
    And we can directly replace the data distributions in the proof of single-channel CPC (see details in Appendix~\ref{app:cpc}) to obtain the inequality below:
    \begin{align}
        \mathcal{L}_{N}^{\text{opt}} &\ge \sum_{i} \left[\mathbb{E}_{X^{i}} \log{\left[ \frac{\frac1n\sum_{j}p(x_{t+k}^{j})}{p(x_{t+k}^{i}|\Phi(c_{t}))} \right]} + \log{N} \right]\\
        &= \mathbb{E}_{X^{1}, X^{2}, \dots, X^{n}} \log{\left[ \frac{[\frac1n\sum_{j}p(x_{t+k}^{j})]^{n}}{\Pi_{j} p(x_{t+k}^{j}|\Phi(c_{t}))} \right]} + n\log{N}.~\label{eq:mcpc_inequ}
    \end{align}
    According to the Jensen Inequality and concavity of the logarithmic function, we obtain that $(\sum_{j}\log{p(x_{t+k}^{j})})/n \le \log{(\sum_{j}p(x_{t+k}^{j})/n)}$. By exponentiating the two equations, we have 
    \begin{equation}
        \Pi_{j} p(x_{t+k}^{j}) \le [\frac1n\sum_{j}p(x_{t+k}^{j})]^{n}.~\label{eq:jensen}
    \end{equation}
    With the help of~\eqref{eq:jensen}, we can further obtain the lower bound of~\eqref{eq:mcpc_inequ}:
    \begin{align}
        \mathcal{L}_{N}^{\text{opt}} &\ge \mathbb{E}_{X^{1}, X^{2}, \dots, X^{n}} \log{\left[ \frac{\Pi_{j}p(x_{t+k}^{j})}{\Pi_{j} p(x_{t+k}^{j}|\Phi(c_{t}))} \right]} + n\log{N} \\
        &= \sum_{i} \left[ -I(x_{t+k}^{i};\Phi(c_{t})) + \log{N} \right].
    \end{align}
    Then we complete the proof.
\end{proof}

We next analyze the advantages of multi-channel CPC over single-channel CPC. 
Our loss function $\mathcal{L}_{N}$ leads to a better predictive representation because we approximate a more informative objective $I(x_{t+k}^{i};\Phi(c_{t}))$, if the optimal loss function for each channel has $\log{N}$ gap with $I(x_{t+k}^{i};\Phi(c_{t}))$, which is the same in single-channel CPC.
Moreover, with the same GPU memory, the more channels, the smaller the batch size that can be accommodated. But we can randomly sample negative samples across all channels, which increases the diversity of negative samples.
However, in order to narrow the approximation gap,~\eqref{eq:jensen} should be considered. The equality sign in this inequality holds if and only if samples from each channel follows the same distribution.
In fact, for many large-scale time series data (\emph{e.g.}, brain signal data used in this work), 
by normalizing each channel, they all exhibit close normal distributions, leading to small gaps in~\eqref{eq:jensen}.

\section{Proposed Method}

In this section, we introduce the details of the novel self-supervised learning framework \model. 
For the commonality between EEG and SEEG, we are inspired by the synergistic effect of brain function and nerve cells, that is, different connectivity patterns correspond to different brain states~\citep{lynn2019physics}. In particular, for brain signals, nerve cells will spontaneously generate traveling waves and spread them out~\citep{davis2020spontaneous}, maintaining some characteristics such as shape
during the process. Therefore,
the degree of channel similarity implies different propagation patterns of traveling waves, reflecting the differences in connectivity patterns to some extent.
Both EEG and SEEG data follow the inherent physiological mechanism.
Therefore, we propose to extract the correlation graph structure between channels (brain areas) as the cornerstone to unify EEG and SEEG (Section~\ref{sec:correlation}).
Next, we introduce three SSL tasks to model brain signals in Section~\ref{sec:ssl_task}. We propose \textit{instantaneous time shift task} based on multi-channel CPC and \textit{delayed time shift task} to capture the spatial and temporal correlation patterns. Then \textit{Replace discriminative task} is designed to preserve characteristics of each channel.

\vpara{Notations.}
For both EEG and SEEG data, there are multiple electrodes with $\mathbf{C}$ channels. We use $X = \{ x_{l} \in \mathbb{R}^{\mathbf{C}}, l=1,\dots,\mathbf{L}\}$ to represent raw time series data with $\mathbf{L}$ time points. $i$ and $j$ denote the index of channels. $Y_{l, i} \in \{ 0, 1\}$ 
is the label for the $l$-th time point of the $i$-th channel. We use a $\mathbf{W}$-length window with no overlap to obtain the time segments $S = \{ s_{t}, t=1,\dots,|S|\}$ (see details in Appendix~\ref{app:segment}). The label corresponding to the $t$-th time segment of the $i$-th channel is 
denoted as $Y^{s}_{t, i}$.

\subsection{Learning Correlations between Channels}\label{sec:correlation}

\begin{figure}[ht]
  \centering
  \includegraphics[width=\linewidth]{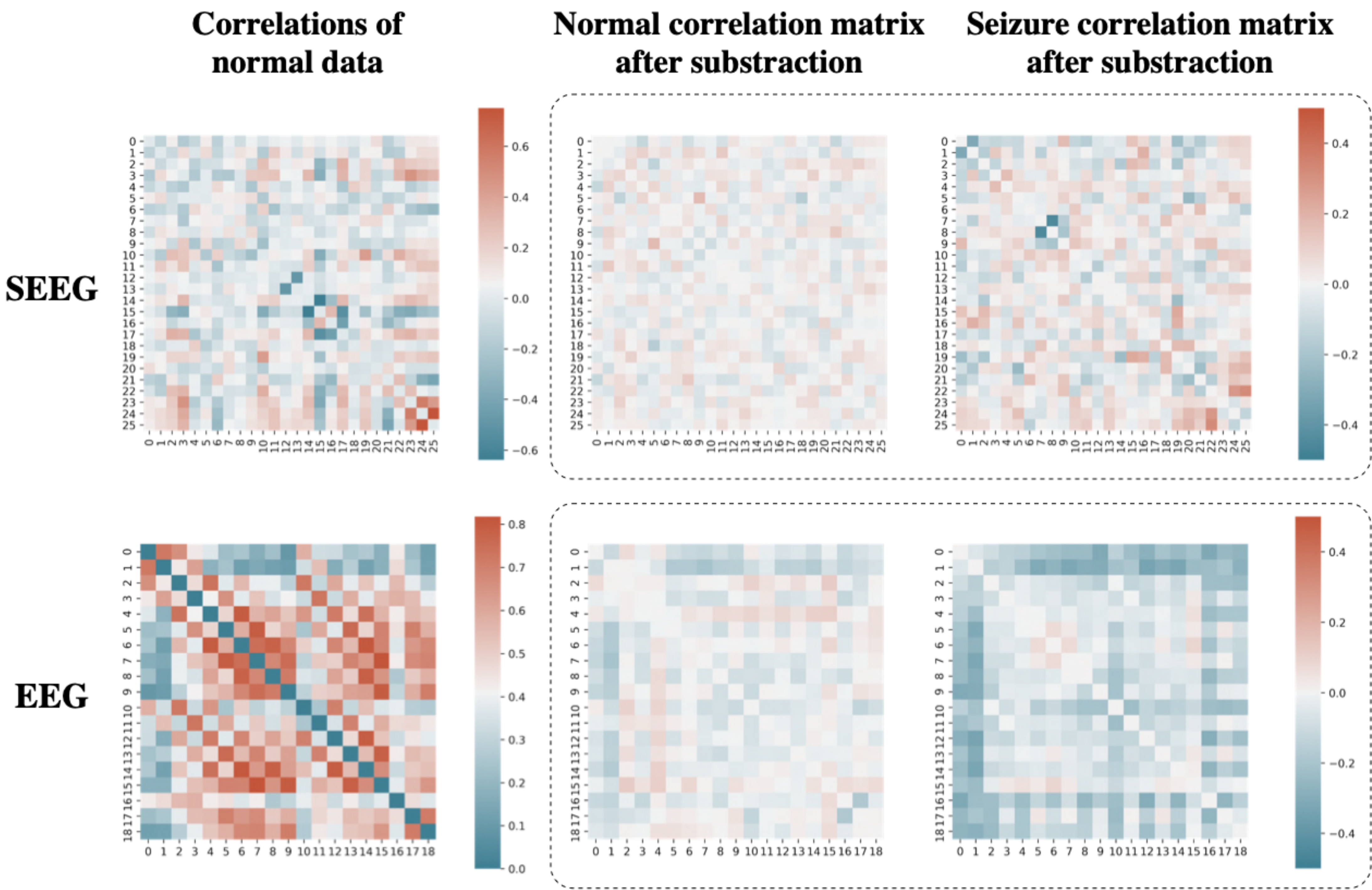}
  \caption{The normal and seizure correlation matrices of EEG and SEEG brain signals. 
  \small
  The top row is for SEEG and the bottom row is for EEG. For clear presentation, we sample some channels in SEEG data. The leftmost two figures are the base correlation matrices on normal data. 
  The two figures in the middle column represent the matrices after subtracting another normal correlation matrices from the base matrices, and the rightmost column includes matrices after subtracting seizure correlation matrices from the base matrices.
  }
  \label{pic:correlation_graph_sub}
\end{figure}

As mentioned above, the correlation patterns between different brain areas can help us to distinguish brain activities in downstream tasks to a large extent.
Taking the seizure detection task as an example, when seizures occur, more rapid and significant propagation of spike-and-wave discharges will appear~\citep{proix2018predicting}, which greatly enhances the correlation between channels. 
This phenomenon is also verified by data observations in Figure~\ref{pic:correlation_graph_sub}.
As Figure~\ref{pic:correlation_graph_sub} shows, for both EEG and SEEG data, we can observe that the correlation matrices are nearly identical on two normal segments without overlap in the same subject. In contrast, the correlation matrix in the epileptic states differs greatly from the normal ones. These data observations verify the conclusion that correlation patterns can help us to distinguish different brain states, and support us to treat correlation graph structure learning as the common cornerstone of our framework.
However, correlations between brain regions are difficult to be observed and recorded directly. 
Therefore, for each time step $t$, our goal is to learn the structure of the correlation graph, whose adjacency matrix is $\mathbf{A}_{t}$, where nodes in the graph indicate channels and weighted edges denote the correlations between channels.

Considering that the brain is in normal and stable state most of the time,
we first define
the \textit{coarse-grained} correlation graph as the prior graph for a particular individual as 
\begin{equation}
    \mathbf{A}^{\text{coarse}}(i, j) = \mathbb{E}_{s_{t}}[\text{Cosine}(s_{t, i}, s_{t, j})],
\end{equation}
where the expectation operation averages over all the correlation matrices computed in only one time segment $s_{t}$, and $\text{Cosine}(\cdot, \cdot)$ denotes the cosine similarity function.

Next, based on $\mathbf{A}^{\text{coarse}}$, for 
each pair of channels, we further model their \textit{fine-grained} short-term correlation within each time segment. 
We assume that the fine-grained correlations follow a Gaussian distribution element-wise, whose location parameters are elements of $\mathbf{A}^{\text{coarse}}$ and scale parameters will be learned from the data. 
By means of the reparameterization trick, the short-term correlation matrix of the $t$-th time segment is sampled from the learned Gaussian distribution:
\begin{align}
    \sigma_{t}(i, j) &= \text{SoftPlus}(\text{MLP}(c^{\text{self}}_{t,\tau,i}, c^{\text{self}}_{t,\tau,j})),\\
    n_{t}(i, j) &\sim \mathcal{N}(0, 1),\\
    \mathbf{A}^{\text{fine}}_{t}(i, j) &= \mathbf{A}^{\text{coarse}}(i, j) + \sigma_{t}(i, j) \times n_{t}(i, j).
\end{align}
$\text{SoftPlus}(\cdot)$ is a commonly used activation function to ensure the learned standard deviation is positive. $c^{\text{self}}_{t,\tau}$ is the contextual representation of raw time segments extracted by encoders (see details in Section~\ref{sec:ssl_task}). 
To remove the spurious correlations caused by low frequency signals and enhance the sparsity, which is a common assumption in neuroscience~\citep{yu2017connectivity}, 
we filter the edges by a threshold-based function to obtain the final correlation graph structure $\mathbf{A}_{t}$: 
\begin{equation}
    \mathbf{A}_{t}(i, j) = \left\{
    \begin{aligned}
    &\mathbf{A}^{\text{fine}}_{t}(i, j), \qquad &\mathbf{A}^{\text{fine}}_{t}(i, j) \geq \theta_{1}, \\ 
    &0, \qquad &\mathbf{A}^{\text{fine}}_{t}(i, j) < \theta_{1}.
    \end{aligned}\right.
\end{equation}

\subsection{Self-supervised Learning for Brain Signals}\label{sec:ssl_task} 

To capture the correlation patterns in space and time, we propose two self-supervised tasks: 
\textit{instantaneous time shift} that is based on multi-channel CPC and captures the short-term correlations focusing on spatial patterns; and \textit{delayed time shift} for temporal patterns in broader time scales. 
\textit{Replace discriminative learning} is designed to preserve the unique characteristics of each channel so as to achieve accurate channel-wise prediction.

\vpara{Instantaneous Time Shift.}
For spatial patterns, we aim to leverage the contextual information of correlated channels to better predict future data of the target channel. Therefore, we apply multi-channel CPC and utilize the fine-grained graph structure $\mathbf{A}_{t}$ obtained in Section~\ref{sec:correlation} as the correlations between channels.

\begin{figure}[ht]
  \centering
    \subfigure[The correlation matrices of delayed time shift of SEEG data.]{
    \begin{minipage}[ht]{\linewidth}
    \centering
    \includegraphics[width=\linewidth]{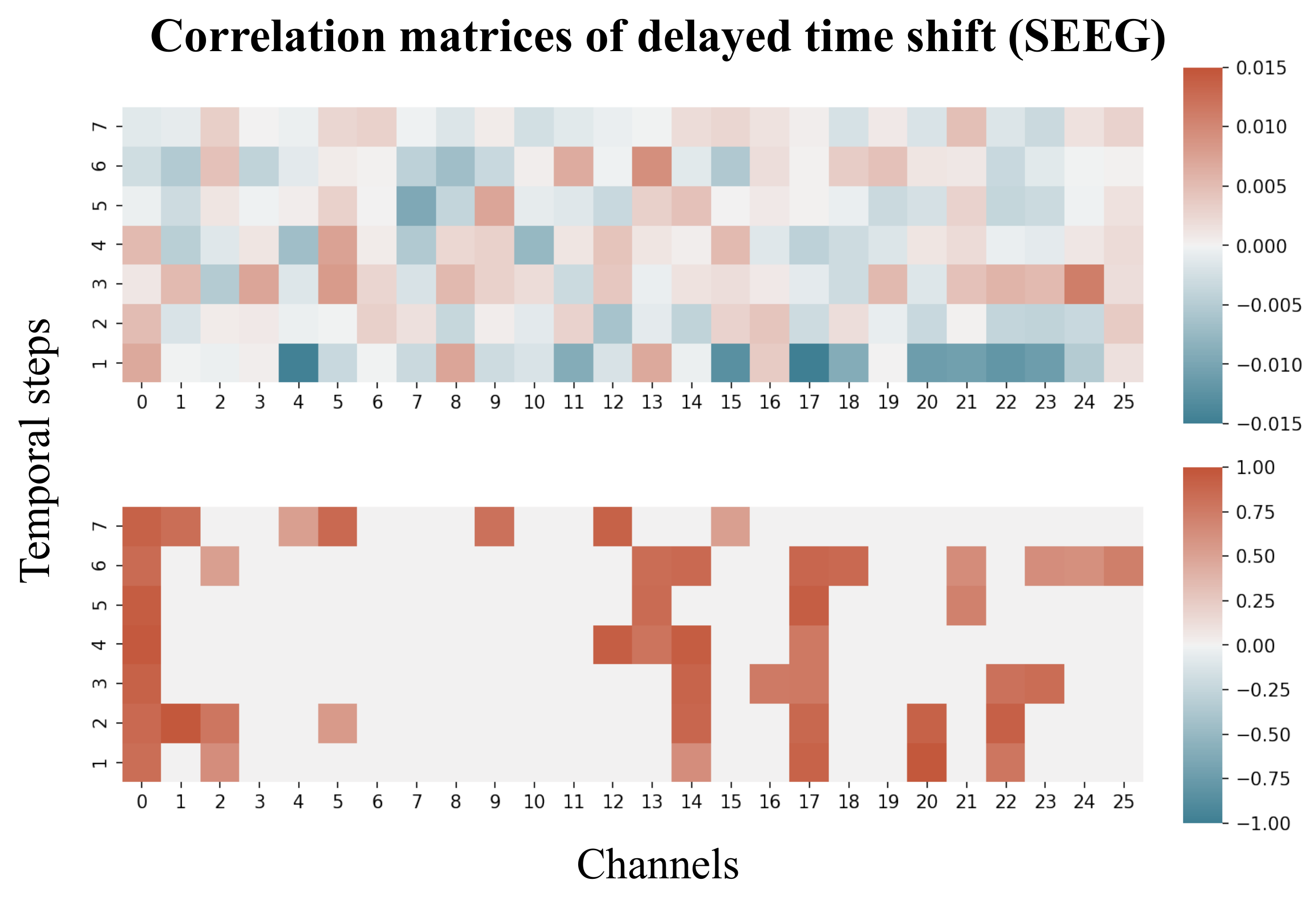}
    \end{minipage}%
    \label{pic:delayed_seeg}}%
    \\
    \subfigure[The correlation matrices of delayed time shift of EEG data.]{
    \begin{minipage}[ht]{\linewidth}
    \centering
    \includegraphics[width=\linewidth]{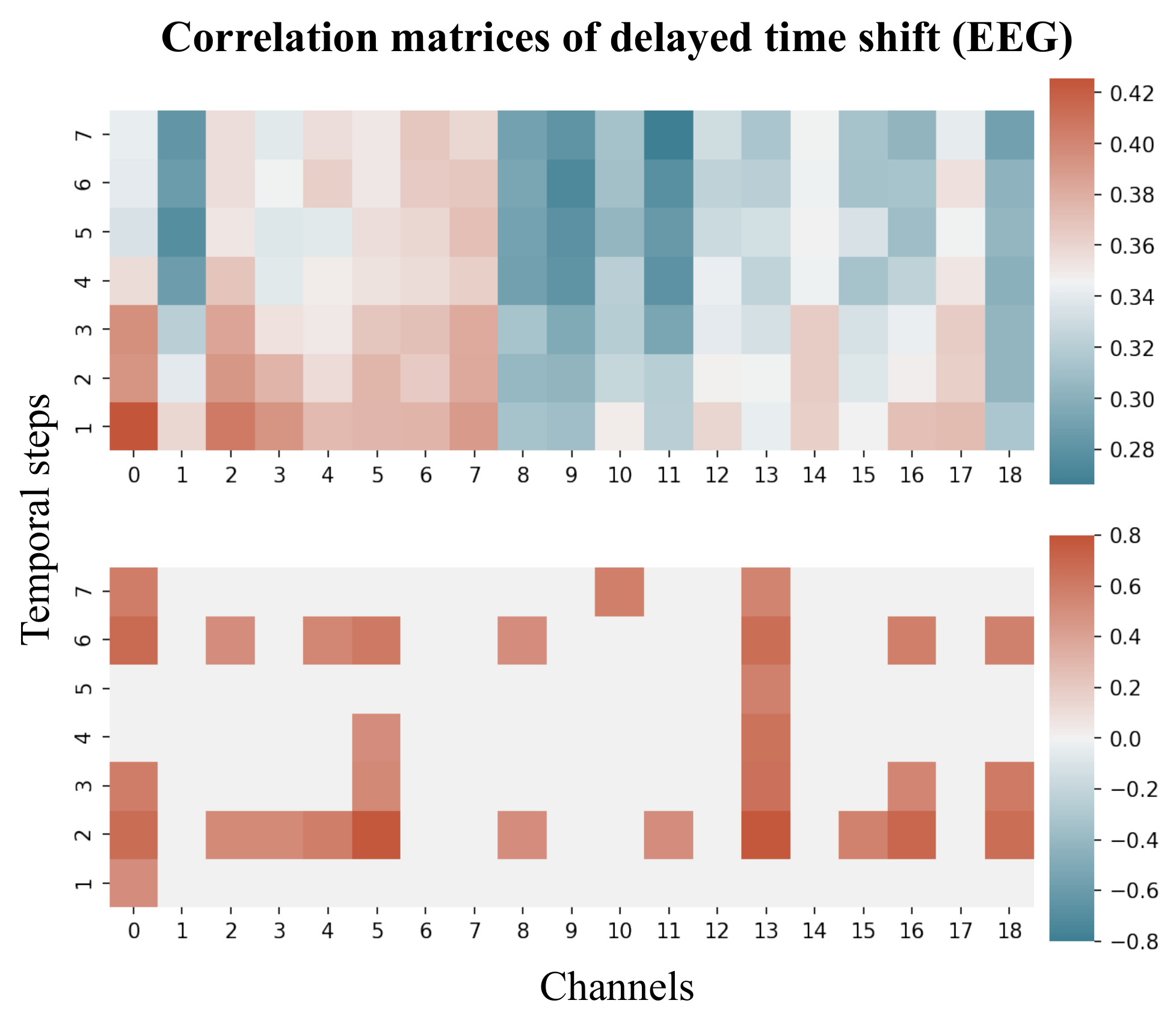}
    \end{minipage}%
    \label{pic:delayed_eeg}}%
  \caption{
  The correlation matrices of delayed time shift of SEEG and EEG. 
  \small
  For each subfigure, the top figure shows the average correlation matrix over all clips.
  And the bottom figure represents the correlation matrix of one particular sampled clip. 
  We compute cosine similarity between the first time segment of the first channel and the time segments of other channels in the next consecutive 7 time steps. For clear presentation, we sample 26 channels for SEEG data and set correlations below 0.5 to 0 for the bottom figure. 
  }
  \label{pic:delayed}
\end{figure}

\begin{figure*}[ht]
  \centering
  \includegraphics[width=\linewidth]{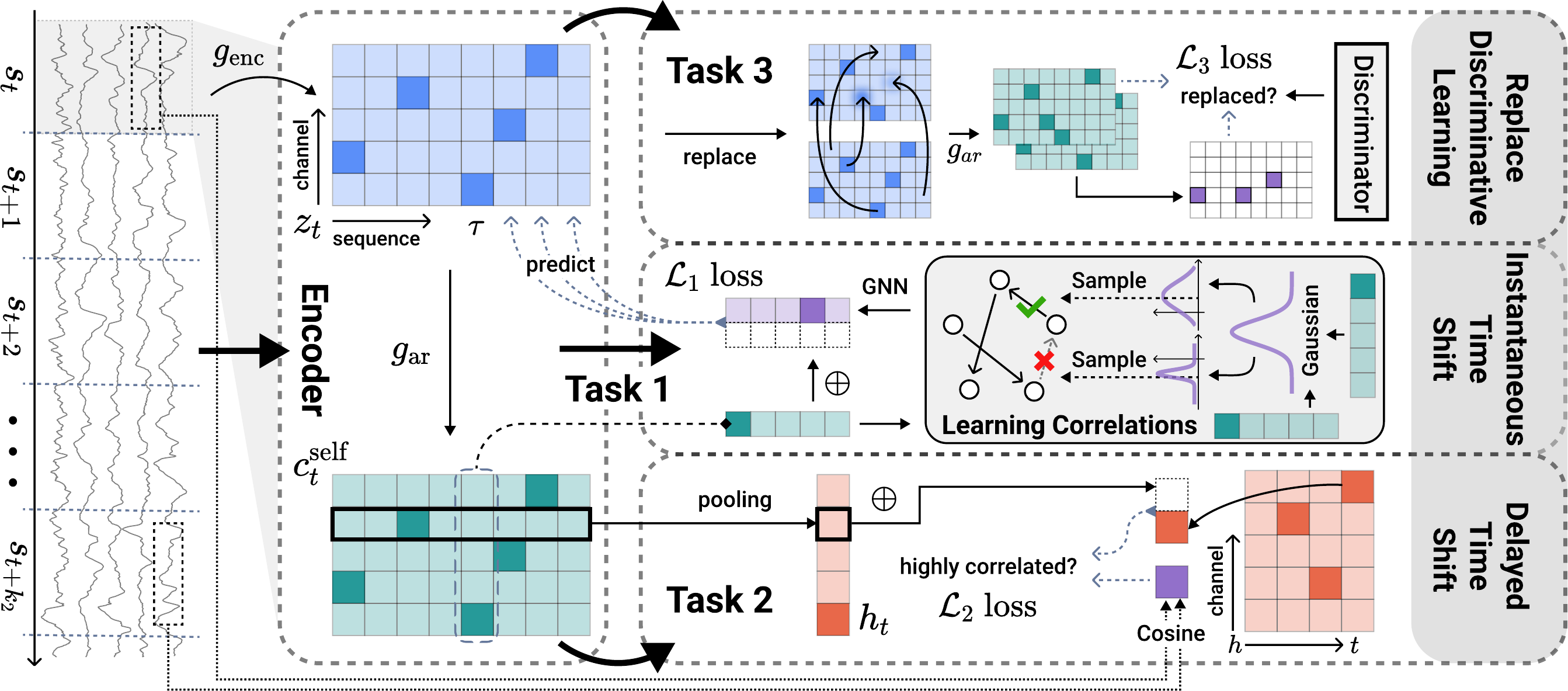}
  \caption{Overview of \model.
  \small
  The leftmost is the raw multi-channel brain signals. We use an encoder to map the raw data into a low-dimensional representation space. To capture the spatial and temporal correlation patterns, we propose three SSL tasks to guide the encoder to learn informative and distinguishable representations.
  }
  \label{pic:model}
\end{figure*}

We first use a non-linear encoder $g_{\text{enc}}$ (1D-CNN with $d$ kernels) mapping the observed time segments to the local latent $d$-dimensional representations $z_{t} = g_{\text{enc}}(s_{t}) \in \mathbb{R}^{\mathcal{T} \times \mathbf{C} \times d}$ for each channel separately. $\mathcal{T}$ is the sequential length after down sampling by $g_{\text{enc}}$. Then an autoregressive model $g_{\text{ar}}$ is utilized to summarize the historical $\tau$-length local information of each channel itself to obtain the respective contextual representations:
\begin{equation}
    c^{\text{self}}_{t, \tau} = g_{\text{ar}}(z_{t, 1}, \cdots, z_{t, \tau}).
\end{equation}
In this step, we only extract the contextual information of all channels independently. Based on the graph structure $\mathbf{A}_{t}$, we instantiate the aggregate function $\Phi(\cdot)$ in~\eqref{eq:mcpc_inequ} as GNNs due to their natural message-passing ability on a graph. Here we use a one-layer directed GCN~\citep{Yun2019gcn} to show the process:
\begin{equation}
    c^{\text{other}}_{t, \tau, i} = \text{ReLU}\left( \frac{\sum_{j \neq i}\mathbf{A}_{t}(i, j) \cdot c^{\text{self}}_{t, \tau, j}}{\sum_{j \neq i}\mathbf{A}_{t}(i, j)} \cdot \Theta \right),
\end{equation}
where $\Theta$ is the learnable matrix. Considering that we only aggregate other channels' information, the self-loop in GCN is removed here. Finally, by combining both $c^{\text{self}}_{t, \tau}$ and $c^{\text{other}}_{t, \tau}$ to obtain the global representations $c_{t, \tau}$, the model can predict the local representations $k_{1}$-step away $z_{t, \tau + k_{1}}$ based on the multi-channel CPC loss:
\begin{align}
    c_{t, \tau} &= \text{Concat}(c^{\text{self}}_{t, \tau}, c^{\text{other}}_{t, \tau}),\\
    \mathcal{L}_{1} = \mathcal{L}_{N} &= - \mathbb{E}_{t, i, k_{1}}\left[ \log \frac{c_{t, \tau, i}^{\top} W_{k_{1}} z_{t, \tau + k_{1}, i}}{\sum_{z_{j} \in X_{t}^{i}} c_{t, \tau, i}^{\top} W_{k_{1}} z_{j}} \right], \label{eq:infonce}
\end{align}
where $X_{t}^{i}$ denotes the random noise set including one positive sample $z_{t, \tau + k_{1}, i}$ and $N-1$ negative samples. $W_{k_{1}}$ is the learnable bilinear score matrix of the $k_{1}$-th step prediction.

\vpara{Delayed Time Shift.}
For brain areas far apart, there exists delayed brain signal propagation, which is confirmed by the data observations showed in Figure~\ref{pic:delayed}.
Figure~\ref{pic:delayed} confirms that there still exist significant correlations between time segments across several time steps. Unlike instantaneous time shift, delayed correlations are not stable. 
This can be concluded from the numerical difference between the averaged correlation matrix and the sampled correlation matrix in both figures.
Therefore, we design a more flexible self-supervised task to learn the delayed correlations.

Our motivation is that if a simple classifier can easily predict whether two time segments are highly correlated, the segment representations will be significantly different from those with weaker correlations. We thus define the delayed time shift task to encourage more distinguishable segment representations.
Similar with instantaneous time shift, we first compute the cosine similarity matrix based on raw data between time segments across several time steps.
For the $t$-th time segment of the $i$-th channel, the long-term correlation matrix $\mathbf{B}^{i}_{t}$ is computed as 
\begin{equation}
    \mathbf{B}^{i}_{t}(k_{2}, j) = \text{Cosine}(s_{t, i}, s_{t + k_{2}, j}),
\end{equation}
where 
$j$ traverses all channels including the $i$-th target channel
and $k_{2}$ traverses at most $K_{2}$ prediction steps. 
Then we construct pseudo labels $Y^{i}_{t}$ according to $\mathbf{B}^{i}_{t}$ to encourage the segment representations with higher correlations to be closer. 
A predefined threshold $\theta_{2}$ is set to assign pseudo labels:
\begin{equation}
    Y^{i}_{t}(k_{2}, j) = \left\{
    \begin{aligned}
    &1, \qquad &\mathbf{B}^{i}_{t}(k_{2}, j) \geq \theta_{2}, \\ 
    &0, \qquad &\mathbf{B}^{i}_{t}(k_{2}, j) < \theta_{2}.
    \end{aligned}\right.
\end{equation}
With the pseudo labels, we define the cross entropy loss of the delayed time shift prediction task:
\begin{gather}
    h_{t} = \text{Pooling}(c^{\text{self}}_{t, 1}, \cdots, c^{\text{self}}_{t, \mathcal{T}}), ~\label{eq:segrepre}\\
    \quad\ \ \hat{p} = \text{Softmax}(\text{MLP}(\text{Concat}(h_{t, i}, h_{t+k_{2}, j}))),\\
    \mathcal{L}_{2} = -\mathbb{E}_{t, i, k_{2}, j} \left[ Y^{i}_{t}(k_{2}, j) \log \hat{p} + (1 - Y^{i}_{t}(k_{2}, j)) \log (1 - \hat{p}) \right] \label{eq:longterm}
\end{gather}
where $\hat{p}$ is the predicted probability that the two segments are highly correlated. In practical application, we randomly choose $50\%$ labels from each $Y^{i}_{t}$ for efficient training.

\vpara{Replace Discriminative Learning.}
Consistently exploiting correlation for all channels will weaken the specificity between channels. However, there are significant differences in the physiological signal patterns of different brain areas recorded by channels. Therefore, retaining the characteristics of each channel cannot be ignored for the modeling of brain signals. 
For this purpose, we further design the replace discriminative learning task.

Following BERT~\citep{devlin2018bert}, we randomly replace $r\%$ local representations throughout $z_{t}$ by $\hat{z}_{t}$, which is sampled from any $\mathcal{T}$ sequences and any $\mathbf{C}$ channels in $z_{t}$.
We use the notation $\mathcal{I}(\hat{z}_{t})$ to represent the new local representations after replacement and the corresponding channel indexes of $\hat{z}_{t}$ in the original sequence. We generate pseudo labels $Y_{t}$ of the task as below:
\begin{equation}
    Y_{t}(\tau, i) = \left\{
    \begin{aligned}
    &1, \qquad &\mathcal{I}(\hat{z}_{t, \tau, i}) \neq i, \\ 
    &0,\qquad &\mathcal{I}(\hat{z}_{t, \tau, i}) = i.
    \end{aligned}\right.
\end{equation}
$\tau$ and $i$ traverse $\mathcal{T}$ sequences and $\mathbf{C}$ channels of $\hat{z}_{t}$.
After obtaining $\hat{z}_{t}$, we put it into the autoregressive model to get the new contextual representations $\hat{c}_{t} = g_{\text{ar}}(\hat{z}_{t})$. Finally, a simple discriminator implemented by an MLP is utilized to classify whether $\hat{c}_{t}$ are replaced by other channels or not:
\begin{equation}
    \mathcal{L}_{3} = -\mathbb{E}_{t, \tau, i}\left[ Y_{t}(\tau, i) \log \hat{q} + (1 - Y_{t}(\tau, i)) \log (1 - \hat{q}) \right], \label{eq:replace}
\end{equation}
where $\hat{q}$ is the predicted probability that $\hat{c}_{t, \tau, i}$ is replaced. When the accuracy of discrimination increases, different channel representations output by the autoregressive model are easier to distinguish. Therefore, the task encourages the model to preserve the unique characteristics of each channel.

Combining the multi-task loss functions~\eqref{eq:infonce},~\eqref{eq:longterm} and~\eqref{eq:replace}, we jointly train \model with $\mathcal{L} = (1 - \lambda_1 - \lambda_2)\mathcal{L}_{1} + \lambda_1\mathcal{L}_{2} + \lambda_2\mathcal{L}_{3}$.
After the SSL stage, the segment representations $h_{t}$ obtained from~\eqref{eq:segrepre} are used for downstream tasks.

\section{Experiments}\label{sec:exp}

\subsection{Datasets and Baselines}\label{sec:dataset}

\vpara{SEEG dataset.}
The SEEG dataset used in our experiment is anonymous and provided by a first-class hospital we cooperate with. For a subject suffering from epilepsy, 4 to 10 invasive electrodes with 52 to 124 channels are used for recording signals.
It is worth noting that since SEEG data are collected in a high frequency (1,000Hz or 2,000Hz) through multiple channels for several days, our data is massive. In total, we have collected 470 hours of SEEG signals with a total capacity of 550GB.
Professional neurosurgeons help us label the epileptic segments for \textbf{each channel}.

We obtain the samples for each subject respectively.
For the $i$-th subject, we first sample a dataset for self-supervised learning which is denoted as $SS_i$ (80\% for training and 20\% for validation), then sample training set $T_i$, validation set $V_i$ and testing set $E_i$ for the downstream stage. $SS_i$, $T_i$ and $V_i$ contain 1,000, 800 and 200 10-second SEEG clips respectively, while $E_i$ contains 510 10-second SEEG clips with positive-negative sample ratio of 1:50. There is no overlap among the samples of the three sets.
We use a 1-second window to segment each clip without overlap and our target is to make predictions for \textbf{all channels} in each \textbf{1-second segment}.

\vpara{EEG dataset.}
We use the Temple University Hospital EEG Seizure Corpus (TUSZ) v1.5.2~\citep{shah2018temple} as our EEG dataset. It is the largest public EEG seizure database, containing 5,612 EEG recordings, 3,050 annotated seizures from clinical recordings, and eight seizure types. We include 19 EEG channels in the standard 10-20 system. 
We randomly split the official TUSZ train set by subjects into training and validation sets at a ratio of 90/10 for model training and hyperparameter tuning respectively, and we keep out the official TUSZ test set for model evaluation. Therefore, the training, validation and testing sets consist of distinct subjects.
After dividing the dataset by subjects,
we start to sample EEG clips.
For the self-supervised learning, we randomly sample 3,000 12-second unlabeled EEG clips for training and validation, with ratios of 90\% and 10\% respectively. As for the downstream task, we first obtain 3,000 sampled 12-second labeled EEG clips (80\% for training and 20\% for validation). Then, we sample another 3,900 12-second labeled EEG clips with positive-negative sample ratio of 1:10 for testing.
It is worth noting that the labels of EEG data are coarse-grained, which means we only have the label of whether epilepsy occurs in a whole EEG clip.

\vpara{Baselines.}
We compare \model with state-of-the-art models including one supervised classification model \textbf{MiniRocket}~\citep{dempster2021minirocket} and several self-supervised and unsupervised models: \textbf{CPC}~\citep{van2018representation}, \textbf{SimCLR}~\citep{chen2020simple}, \textbf{Triplet-Loss (T-Loss)}~\citep{franceschi2019unsupervised}, \textbf{Time Series Transformer (TST)}~\citep{zerveas2021transformer}, \textbf{GTS}~\citep{shang2021discrete}, \textbf{TS-TCC}~\citep{eldele2021time} and \textbf{TS2Vec}~\citep{yue2021ts2vec}.

\subsection{Experimental Setup}\label{sec:setup}

For EEG data, as the number of subjects is large while the number of samples for each subject is very small, we follow the standard experimental setting to divide the training, validation and testing sets by subjects.
As for SEEG data, since every subject includes many samples, it is accessible to sample training, validation and testing sets for each subject respectively.
To demonstrate the effectiveness of \model, we first formally define the seizure detection task. Then we perform three experiments to show that our model outperforms the state-of-the-art baselines and has the ability to be deployed to clinical practice. 
We also show the ablation study and case study of the correlation graph in Section~\ref{subsec:ablation} and~\ref{subsec:case}.
The hyperparameter analysis is showed in Appendix~\ref{app:hyper}.
We report the results of another downstream task of emotion recognition in Appendix~\ref{app:emotion}.
In order to ensure the reliability of the experimental results, we repeat all the experiments five times with five different random seeds in the fine-tuning stage and report standard deviation in all tables.

\begin{task}[Seizure Detection]
    \textbf{Given} a time-ordered set including $I_{\mathcal{S}}$ consecutive time segments with the index of the first segment being $t_{0}$: $\mathcal{S} = \{ s_{t_{0}}, \dots, s_{t_{0} + I_{\mathcal{S}}} \}$,
    models \textbf{predict} the labels $\hat{Y}^{s}_{t, i}$ for all time segments in $\mathcal{S}$ (\emph{i.e.}, $t = t_{0}, \dots, t_{0} + I_{\mathcal{S}}$) and all channels in each segment (\emph{i.e.}, $i = 1, \dots, \mathbf{C}$).
\end{task}

\vpara{Subject dependent experiment~\citep{chen2022brainnet}.}
Due to the larger difference between subjects in SEEG dataset than that in EEG dataset, we first perform the subject dependent experiment to obtain the upper bound of model performance on SEEG dataset. 
More specifically, for the $i$-th subject, we first perform self-supervised learning of the model on unlabeled data sampled from itself (\emph{i.e.}, $SS_i$).
When training the downstream task, the encoder of SSL models will be \textbf{fine-tuned with a very low learning rate} on labeled data sampled from itself (\emph{i.e.}, $T_i$ and $V_i$).
Finally, we test the models on $E_i$ and report the average performance over all subjects.
For fair comparison, we use \textbf{the same downstream model and experimental setup} for all models (see details in Appendix~\ref{app:main_imp}).

\vpara{Subject independent experiments.}
To meet practical clinical needs, we design two clinically feasible experiments. \textbf{The first is the domain generalization experiment}, that is, training the model on data of existing subjects and directly predicting data of unknown subjects. This is the standard experimental setting on EEG dataset. 
As for SEEG dataset, we follow the 3-1-1 setting, where 3 subjects are used for training (\emph{i.e.}, SSL on $SS_i, SS_j, SS_k$; fine-tuning on $T_i, T_j, T_k$), 1 subject is used for validation (\emph{i.e.}, $V_m$) and 1 subject is used for testing (\emph{i.e.}, $E_n$). Note that $i$, $j$, $k$, $m$ and $n$ are indexes for different subjects.
We conduct the experiments for random combinations, pick up the best result for each subject, and report the average results over all subjects.

\textbf{The second is the domain adaptation experiment~\citep{motiian2017unified}.}
Different from the ideal domain generalization experiment which does not use the labeled data of target subjects at all, domain adaptation experiment allows using a small amount of the data to achieve better clinical performance of our model.
This is because of the large data size due to the long-time records of the subjects in the SEEG dataset, and even if the model is fine-tuned with partially labeled data, it is clinically valuable to predict the large amount of remaining data in the target subjects.
In this experiment, we first perform SSL on one subject (\emph{i.e.}, source domain $SS_i$) and then fine-tuning is performed using partially labeled data from another subject (\emph{i.e.}, target domain $T_j$ and $V_j$).
Finally, we perform seizure detection on the testing set of the target subject (\emph{i.e.}, $E_j$).
We pick up four subjects with typical seizure patterns in the SEEG dataset, and report the results of all one-to-one combinations.

\begin{table}[ht]
    \caption{The average performance of the subject dependent experiment on SEEG dataset.}
    \label{tab:seizure_prediction_SEEG}
    \setlength\tabcolsep{3.9pt}
    \centering
    \begin{tabular}{lcccc}
        \toprule
        \textbf{Models} & Pre. & Rec. & $F_1$ & $F_2$ \\
        \midrule
        MiniRocket & 22.98\small{$\pm$0.15} & \textbf{66.24\small{$\pm$0.26}} & 31.79\small{$\pm$0.19} & 43.58\small{$\pm$0.22} \\ 
        \midrule
        CPC & 27.65\small{$\pm$4.49} & 55.07\small{$\pm$3.52} & 34.20\small{$\pm$3.40} & 42.73\small{$\pm$2.57} \\
        SimCLR & 11.06\small{$\pm$3.95} & 51.54\small{$\pm$5.87} & 16.60\small{$\pm$4.68} & 25.41\small{$\pm$4.95} \\
        T-Loss & 29.29\small{$\pm$2.65} & 51.55\small{$\pm$2.53} & 36.00\small{$\pm$1.97} & 43.13\small{$\pm$1.57} \\
        TST & 13.60\small{$\pm$3.48} & 44.65\small{$\pm$4.21} & 19.80\small{$\pm$3.73} & 28.41\small{$\pm$3.29} \\
        GTS & 24.29\small{$\pm$4.26} & 40.39\small{$\pm$5.80} & 29.16\small{$\pm$2.97} & 34.17\small{$\pm$2.36} \\
        TS-TCC & 22.10\small{$\pm$7.65} & 49.94\small{$\pm$5.41} & 25.32\small{$\pm$8.02} & 32.74\small{$\pm$7.95} \\
        TS2Vec & 30.56\small{$\pm$2.17} & 52.83\small{$\pm$2.89} & 36.03\small{$\pm$1.72} & 43.35\small{$\pm$1.59} \\
        \midrule
        \model & \textbf{37.97\small{$\pm$2.75}} & 65.07\small{$\pm$2.68} & \textbf{46.45\small{$\pm$2.25}} & \textbf{55.28\small{$\pm$1.77}} \\
        \bottomrule
    \end{tabular}
\end{table}

\subsection{Subject Dependent Experiment}

The average performance of the subject dependent experiment on the SEEG dataset is presented in Table~\ref{tab:seizure_prediction_SEEG}. 
Since the positive-negative sample ratio of SEEG dataset is imbalanced, $F$-score is a more appropriate metric to evaluate the performance of models than only considering precision or recall.
Especially in clinical applications, doctors pay more attention to finding as much seizures as possible, we thus choose $F_{1}$ and $F_{2}$ scores in the experiment.
Overall, \model improves the $F_{1}$-score by 28.92\% and the $F_{2}$-score by 26.85\% on SEEG dataset, compared to the best baseline, demonstrating that \model can learn more informative representations from SEEG data.
Through this experiment, we obtain the upper bound of the performance of models on SEEG dataset. We can find that it is still difficult to achieve high performance even if models are trained, verified and tested on the same subject. Combined with the analysis of subsequent experimental results, this reflects that seizure detection on SEEG data is much more difficult than that on EEG.

\begin{table*}[ht]
    \caption{The average performance of the domain generalization experiment on SEEG and EEG datasets.}
    \label{tab:dg}
    \centering
    \begin{tabular}{lccccccccc}
        \toprule
        \multirow{2}{*}{\textbf{Models}} & \multicolumn{4}{c}{SEEG} & \multicolumn{5}{c}{EEG} \\
        \cmidrule(lr){2-5}
        \cmidrule(lr){6-10}
        & Pre. & Rec. & $F_1$ & $F_2$ & Pre. & Rec. & $F_1$ & $F_2$ & AUROC \\
        \midrule
        MiniRocket & 5.85\small{$\pm$0.20} & \textbf{39.18\small{$\pm$0.59}} & 9.93\small{$\pm$0.29} & 17.24\small{$\pm$0.37} & \textbf{22.86\small{$\pm$0.84}} & 63.08\small{$\pm$1.47} & 33.56\small{$\pm$1.11} & 46.66\small{$\pm$1.33} & 75.30\small{$\pm$0.77} \\
        \midrule
        CPC & 22.88\small{$\pm$5.06} & 23.92\small{$\pm$3.90} & 20.11\small{$\pm$3.27} & 21.23\small{$\pm$2.49} & 22.81\small{$\pm$2.04} & 58.31\small{$\pm$7.55} & 32.50\small{$\pm$1.24} & 44.02\small{$\pm$2.43} & 74.53\small{$\pm$1.00} \\
        SimCLR & 14.02\small{$\pm$3.71} & 26.36\small{$\pm$4.99} & 11.07\small{$\pm$3.49} & 13.47\small{$\pm$4.01} & 12.63\small{$\pm$1.62} & 74.88\small{$\pm$16.77} & 21.33\small{$\pm$1.95} & 36.78\small{$\pm$2.61} & 55.86\small{$\pm$5.36} \\
        T-Loss & 21.38\small{$\pm$4.25} & 28.50\small{$\pm$4.07} & 23.48\small{$\pm$3.30} & 25.90\small{$\pm$3.06} & 20.72\small{$\pm$1.26} & 69.25\small{$\pm$3.99} & 31.82\small{$\pm$1.08} & 47.00\small{$\pm$0.50} & 75.88\small{$\pm$0.49} \\
        TST & 8.37\small{$\pm$3.96} & 32.48\small{$\pm$8.25} & 11.80\small{$\pm$3.91} & 15.67\small{$\pm$3.69} & 15.65\small{$\pm$1.54} & 28.59\small{$\pm$12.93} & 19.65\small{$\pm$4.36} & 23.87\small{$\pm$8.09} & 58.20\small{$\pm$4.27} \\
        GTS & 24.16\small{$\pm$5.91} & 27.99\small{$\pm$4.98} & 22.77\small{$\pm$2.69} & 24.15\small{$\pm$2.79} & 18.86\small{$\pm$1.09} & 62.51\small{$\pm$5.04} & 28.88\small{$\pm$0.88} & 42.54\small{$\pm$1.48} & 71.69\small{$\pm$1.88} \\
        TS-TCC & 24.24\small{$\pm$4.51} & 26.61\small{$\pm$5.96} & 19.89\small{$\pm$5.23} & 22.11\small{$\pm$5.08} & 15.55\small{$\pm$0.88} & 39.76\small{$\pm$11.08} & 21.89\small{$\pm$1.20} & 29.60\small{$\pm$4.64} & 58.63\small{$\pm$1.62} \\
        TS2Vec & 27.93\small{$\pm$5.23} & 29.49\small{$\pm$3.97} & 26.78\small{$\pm$3.29} & 27.88\small{$\pm$3.52} & 21.40\small{$\pm$0.63} & 58.31\small{$\pm$6.14} & 31.24\small{$\pm$1.18} & 43.24\small{$\pm$2.78} & 73.35\small{$\pm$1.02} \\
        \midrule
        \model & \textbf{30.69\small{$\pm$5.92}} & 38.94\small{$\pm$4.34} & \textbf{32.61\small{$\pm$3.60}} & \textbf{35.64\small{$\pm$3.04}} & 22.13\small{$\pm$1.03} & \textbf{76.99\small{$\pm$4.49}} & \textbf{34.32\small{$\pm$0.90}} & \textbf{51.34\small{$\pm$0.97}} & \textbf{77.96\small{$\pm$0.97}} \\
        \bottomrule
    \end{tabular}
\end{table*}

\subsection{Domain Generalization Experiment}

In this experiment, we validate and compare the generalization ability of all models under a strict setting, in which the models are trained on source subjects and then directly perform seizure detection on the unseen target subjects.
This is an ideal scenario for clinical applications and the results are shown in Table~\ref{tab:dg}.
For SEEG dataset, in general, the performance of models under the domain generalization setting decreases significantly (41.73\% on average in terms of $F_2$-score) compared with that in subject dependent experiment. The drop for recall metric is more pronounced, confirming that the distribution shift of subjects in SEEG data is more significant than that in EEG. This results from the fact that different brain regions and different types of epileptic waves have different physiological properties and patterns.
Nonetheless, \model still improves $F_1$ and $F_2$ scores by 21.77\% and 27.83\% respectively, compared to the best baseline. The results prove that
\model has a superior generalization ability benefiting from rational inductive assumption of model design.
We point out that although GTS is also graph-based model, it directly learns the graph structure for each segment and ignores the stable and long-term correlations between different channels. This implies that our proposed graph structure learning strategy based on the stable correlations is reasonable and effective.

Table~\ref{tab:dg} also shows the results of domain generalization experiment on EEG dataset. Following the common evaluation scheme on EEG dataset~\citep{tang2022selfsupervised}, we add Area Under the Receiver Operating Characteristic (AUROC) metric in our experiment.
Our model is designed to learn the representation for each channel, while there is only one label for an EEG clip. Therefore,
it requires the pooling operation to aggregate representations output by our model over channels and time segments for seizure detection. 
This setting makes the performance improvement of our model not as significant as that in the SEEG experiment.
Nevertheless, \model still outperforms all baselines on $F_{1}$-score, $F_{2}$-score and AUROC with an increase of 2.26\%, 9.23\% and 2.74\%, respectively.
SimCLR gets the highest recall but the lowest precision and AUROC, indicating that it may be not reasonable to regard time segments as independent samples without considering the contextual data. 
The worst performance for TST shows that mask-prediction SSL paradigm may not be suitable for non-stationary time series data.

\begin{table*}[ht]
    \caption{The performance of the domain adaptation experiment on SEEG dataset in terms of $F_2$-score. 
    \small
    DA row denotes the performance of \model in the domain adaptation experiment. Max-base and Non-DA rows represent the best performance of baselines and \model in the subject dependent experiment. We bold the best result and underline the second best result.}
    \label{tab:transfer_learning}
    \setlength\tabcolsep{2pt}
    \centering
    \begin{tabular}{lcccccccccccc}
        \toprule
        \multirow{2}{*}{\textbf{Setting}} & \multicolumn{3}{c}{Group $A$} & \multicolumn{3}{c}{Group $B$} & \multicolumn{3}{c}{Group $C$} & \multicolumn{3}{c}{Group $D$}\\
        \cmidrule(lr){2-4}
        \cmidrule(lr){5-7}
        \cmidrule(lr){8-10}
        \cmidrule(lr){11-13}        
        & $B{\rightarrow}A$ & $C{\rightarrow}A$ & $D{\rightarrow}A$
        & $A{\rightarrow}B$ & $C{\rightarrow}B$ & $D{\rightarrow}B$
        & $A{\rightarrow}C$ & $B{\rightarrow}C$ & $D{\rightarrow}C$
        & $A{\rightarrow}D$ & $B{\rightarrow}D$ & $C{\rightarrow}D$\\
        \midrule
        DA & \small{68.55$\pm$4.27} & \underline{\small{69.14$\pm$6.54}} & \small{68.78$\pm$4.12} & \small{41.08$\pm$2.59} & \small{46.06$\pm$3.05} & \underline{\small{46.12$\pm$2.04}}
        & \small{40.04$\pm$3.98} & \small{39.34$\pm$2.11} & \textbf{\small{48.64$\pm$5.48}} & \underline{\small{80.82$\pm$0.65}} & \small{79.90$\pm$1.11} & \small{80.72$\pm$1.31} \\
        Max-base & \multicolumn{3}{c}{62.49\small{$\pm$2.30}} & \multicolumn{3}{c}{39.78\small{$\pm$2.04}} & \multicolumn{3}{c}{33.59\small{$\pm$2.23}} & \multicolumn{3}{c}{75.35\small{$\pm$0.79}} \\
        Non-DA & \multicolumn{3}{c}{\textbf{70.63\small{$\pm$1.41}}} & \multicolumn{3}{c}{\textbf{46.62\small{$\pm$2.42}}} & \multicolumn{3}{c}{\underline{46.09\small{$\pm$2.35}}} & \multicolumn{3}{c}{\textbf{83.27\small{$\pm$0.95}}} \\
        \bottomrule 
    \end{tabular}
\end{table*}

\subsection{Domain Adaptation Experiment}

According to the results of domain generalization experiment, it is difficult for \model to achieve competitive performance as shown in Table~\ref{tab:seizure_prediction_SEEG} on SEEG dataset. 
The results show that seizure detection on SEEG dataset is much more difficult than that on EEG dataset.
Alternatively, due to the long-time record, clinical SEEG data contains tens or even hundreds of seizures, allowing us to use a small amount of labeled data to fine-tune our model and then use it to predict the remaining data. 
In this way, \model can still achieve great performance, showing good generalization ability and clinical application value of our model.
Table~\ref{tab:transfer_learning} shows the performance of the domain adaptation (DA) experiment for four subjects with typical seizure patterns provided by doctors from SEEG dataset. 
More specifically, we train \model on one subject and fine-tune it on all other three subjects.
$B{\rightarrow}A$ denotes that the SSL model is trained on Subject-B, and then fine-tuned and tested on data from Subject-A.
The results of Max-base and Non-DA rows correspond to the performance of the best baseline and \model respectively in scenarios $A{\rightarrow}A$, $B{\rightarrow}B$, $C{\rightarrow}C$ and $D{\rightarrow}D$.

Compared with the results of the setting that the self-supervised model and downstream model are both trained on the same subject, the $F_{2}$-scores of all 12 cross-domain scenarios reduce by less than 15\%. 
Additionally, it can be observed that in all cross-domain scenarios, \model beats the best baseline in the corresponding scenarios without DA. 
It is worth noting that $D{\rightarrow}C$ scenario even outperforms corresponding Non-DA result. The possible reason is that the signal patterns on Subject-D are more significant and recognizable than those on Subject-C.
Therefore, the SSL model trained on higher quality source domain can better distinguish signal states when performing downstream tasks on target domain.
Overall, the domain adaptation experiment makes \model achieve competitive performance as shown in  Table~\ref{tab:seizure_prediction_SEEG} by fine-tuning it on only a small amount of labeled data from the target domain. 
The results suggest that \model captures the inherent features and outputs generalized representations between subjects, because we fine-tune the SSL model with a very low learning rate (1e-6). 
From the perspective of pre-training, the SSL model trained on the source subject gives good initial parameters for the fine-tuning stage on the target subject.

\subsection{Ablation Study}\label{subsec:ablation}

Considering the complexity of our model, we conduct sufficient ablation experiments to demonstrate the effectiveness of each component in \model.
Specifically, we mainly compare \model with three types of different model variants.
\begin{itemize}[leftmargin=*]
    \item[(1)] \textbf{Replace the method to aggregate channel information.}
    To verify the effectiveness of our proposed graph structure learning, we have proposed two ideas on how to directly implement the multi-channel CPC in Section~\ref{sec:theory}. For the second idea, we have reported the results of a shared CPC regarding all channels as one on the CPC row of Table~\ref{tab:seizure_prediction_SEEG}. For the first idea, we design two strategies to combine multi-channel CNN or MLP into CPC respectively to learn representations for each channel. See detailed description in Appendix~\ref{app:ablation}.

    \item[(2)] \textbf{Remove one component.}
    We firstly remove the correlation graph structure learning module from the instantaneous time shift task and degenerate the task to single-channel CPC while still uniformly sampling negative samples in all channels.
    This variant is denoted as {\model\small{-Graph}}. 
    Next, we respectively remove the whole instantaneous time shift task, the delayed time shift task and replace discriminative task. These variants are denoted as {\model\small{-Instant}}, {\model\small{-Delay}} and {\model\small{-Replace}}.

    \item[(3)] \textbf{Preserve one SSL task.}
    {\model\small{-onlyInstant}}, {\model\small{-onlyDelay}} and {\model\small{-onlyReplace}} indicate that \model only performs instantaneous time shift task, delayed time shift task and replace discriminative task respectively.
    \end{itemize}

\begin{table}[ht]
    \caption{The results of ablation study.}
    \label{tab:ablation_study}
    \setlength\tabcolsep{1.2pt}
    \centering
    \begin{tabular}{lcccc}
        \toprule
        \textbf{Models} & Pre. & Rec. & $F_1$ & $F_2$ \\
        \midrule
        CPC & 27.65\small{$\pm$4.49} & 55.07\small{$\pm$3.52} & 34.20\small{$\pm$3.40} & 42.73\small{$\pm$2.57} \\
        {CPC\small{-Conv}} & 6.39\small{$\pm$0.77} & 33.21\small{$\pm$4.00} & 10.53\small{$\pm$1.07} & 17.46\small{$\pm$1.42} \\
        {CPC\small{-MLP}} & 25.84\small{$\pm$3.07} & 52.70\small{$\pm$3.65} & 32.18\small{$\pm$2.46} & 40.34\small{$\pm$2.05} \\
        \midrule
        {\model\small{-Graph}} & 36.72\small{$\pm$4.59} & 60.48\small{$\pm$4.47} & 43.61\small{$\pm$3.08} & 51.47\small{$\pm$2.68} \\
        {\model\small{-Instant}} & 34.49\small{$\pm$4.37} & 55.41\small{$\pm$3.90} & 41.57\small{$\pm$3.48} & 48.38\small{$\pm$2.52} \\
        {\model\small{-Delay}} & 35.00\small{$\pm$4.49} & \textbf{65.61\small{$\pm$2.94}} & 42.97\small{$\pm$3.61} & 52.51\small{$\pm$1.93} \\
        {\model\small{-Replace}} & 36.08\small{$\pm$5.35} & 63.67\small{$\pm$4.24} & 43.66\small{$\pm$3.66} & 52.49\small{$\pm$2.32} \\
        \midrule
        {\model\small{-onlyInstant}} & 36.43\small{$\pm$4.44} & 63.66\small{$\pm$2.12} & 43.35\small{$\pm$3.83} & 51.82\small{$\pm$2.67} \\
        {\model\small{-onlyDelay}} & 31.59\small{$\pm$4.24} & 55.03\small{$\pm$5.26} & 38.56\small{$\pm$2.84} & 46.05\small{$\pm$2.26} \\
        {\model\small{-onlyReplace}} & 34.13\small{$\pm$6.84} & 56.06\small{$\pm$3.68} & 40.02\small{$\pm$4.47} & 47.44\small{$\pm$2.40} \\
        \midrule
        \model & \textbf{37.97\small{$\pm$2.75}} & 65.07\small{$\pm$2.68} & \textbf{46.45\small{$\pm$2.25}} & \textbf{55.28\small{$\pm$1.77}} \\
        \bottomrule
    \end{tabular}
\end{table}

Table~\ref{tab:ablation_study} shows the results of ablation study on SEEG dataset. It can be observed that the complete \model achieves the best performance on $F_1$ and $F_2$ scores, demonstrating the effectiveness of each component in our model design. 
For the first type of variants, we can observe that the performance of {CPC\small{-Conv}} decreases dramatically. We speculate that this is because the channels are relatively independent, and the correlation between most channels is weak or even non-existent. Direct adoption of multi-channel convolution may introduce spurious and noisy correlations. However, the graph structure learning proposed by us has a sparsity assumption, and the representation extraction of each channel is relatively independent, so it can effectively learn and aggregate more significant information. 
For {CPC\small{-MLP}}, we use an MLP to aggregate the representations of other channels, and then concatenate it with the representation of the target channel to predict future data. Unlike {CPC\small{-Conv}}, which adopts multi-channel convolution for the raw data to obtain the \textit{mixed} low-level representations, {CPC\small{-MLP}}, like \model, learns the correlation of channels based on the \textit{separate} high-level representations. Therefore, the performance of {CPC\small{-MLP}} does not drop as dramatically as that of {CPC\small{-Conv}}.

For {\model\small{-Instant}}, the significant decrease in performance illustrates that capturing the spatial and short-term patterns is quite important and is the key to learning the essential representations in multi-channel brain signals.
For {\model\small{-Graph}}, the decrease in performance demonstrates that multi-channel CPC can greatly help learn more informative representations.
Additionally, the performance in {\model\small{-Delay}} and {\model\small{-Replace}} also decreases significantly, illustrating that modeling long-term temporal patterns and preserving the characteristics of channels can help learn more distinguishable representations.
For the third type of variants,
it can be observed that the instantaneous time shift is the most important task, and the delayed time shift task and the replace discriminative task contribute similarly to the performance of the complete model.

\subsection{Case Study}\label{subsec:case}

\begin{figure}[ht]
  \centering
    \subfigure[Normal correlation graph.]{
    \begin{minipage}[ht]{0.5\linewidth}
    \centering
    \includegraphics[width=4.3cm]{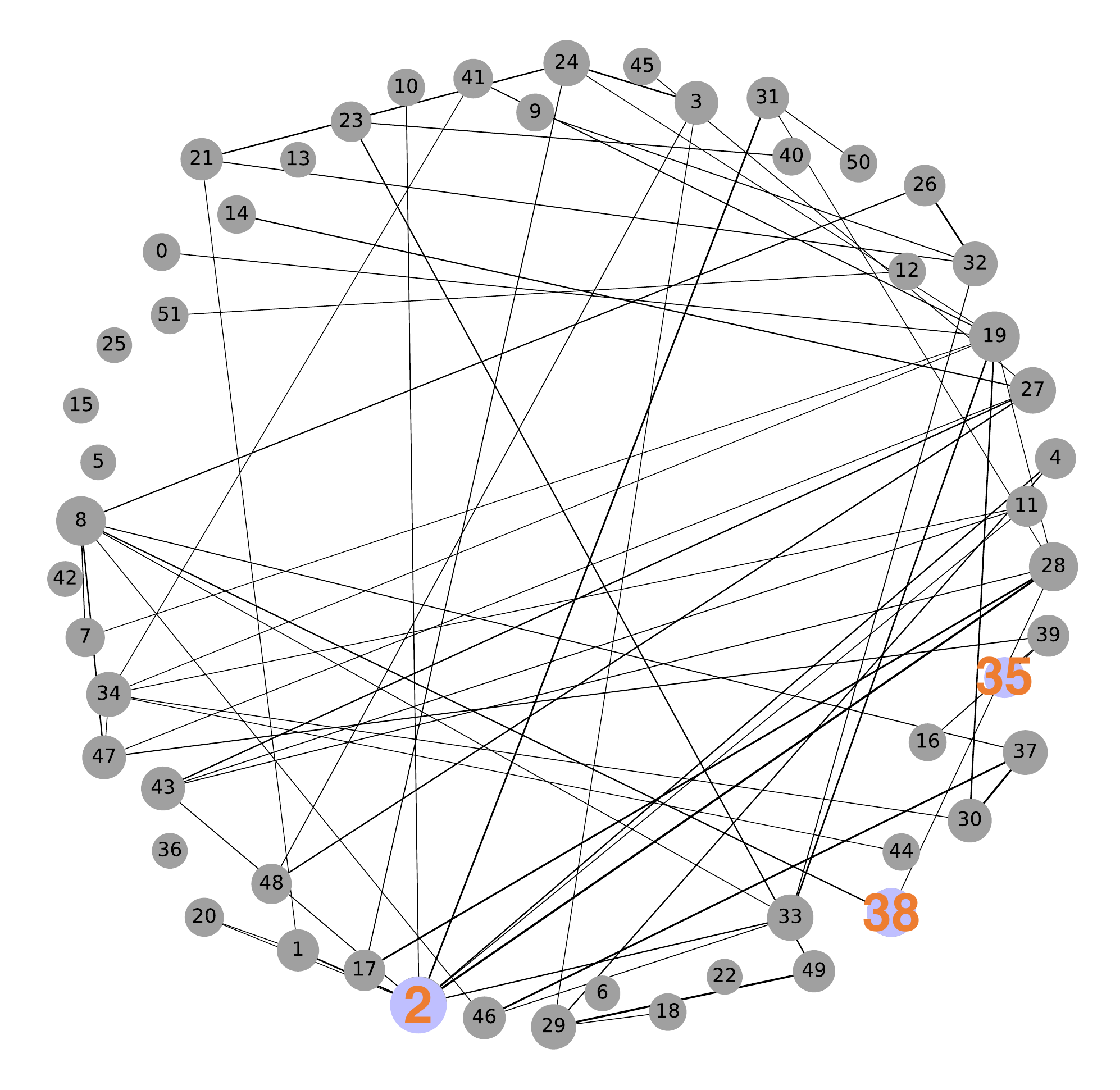}
    \end{minipage}%
    \label{subpic:normal_corGraph}}%
    \subfigure[Seizure correlation graph.]{
    \begin{minipage}[ht]{0.5\linewidth}
    \centering
    \includegraphics[width=4.3cm]{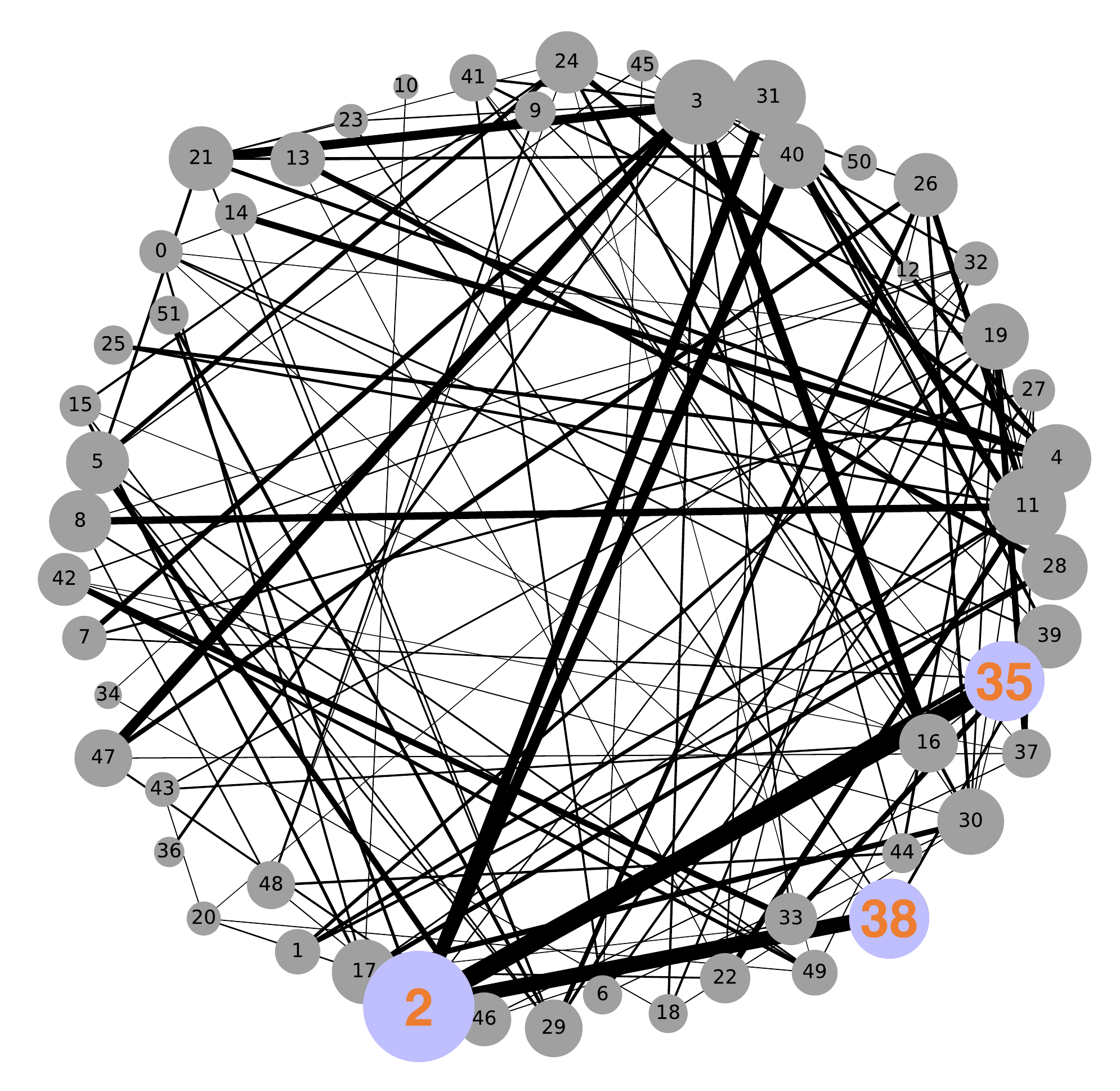}
    \end{minipage}%
    \label{subpic:seizure_corGraph}}%
  \caption{
  Case study on correlation graphs learned by \model.
  }
  \label{pic:corGraph}
\end{figure}

In this section, we study the correlation graphs between the channels learned by \model. We randomly sample normal and seizure SEEG clips of one particular subject, and visualize their correlation graphs $\mathbf{A}_{t}$ (defined in Section~\ref{sec:correlation}) in Figure~\ref{pic:corGraph}.
In this figure, the thickness of an edge indicates its weight. 
And the larger the sum of weights of the edges connected to the node, the larger the size of the circle of the node.
It can be observed that in the normal state, the correlation is sparser and the weights for edges are smaller, indicating a holistically weaker correlation between channels. 
In contrast, in the seizure state, the connection pattern between channels varies, where the correlation becomes denser and the edge weights become larger. Furthermore, in Figure~\ref{subpic:seizure_corGraph}, edges with larger weights are usually connected to 2 seizure channels. For example, Channel-2, Channel-35 and Channel-38 are all in seizure states and the edge weights between them are large, indicating that the brain areas recorded by the three channels have a higher probability of being the focal area. This can help neurosurgeons to better localize seizure lesions.

\section{Conclusion} \label{sec:conclusion}

In this paper, we propose a general multi-channel SSL framework \model, which can be applied for learning representations of both EEG and SEEG brain signals. 
Based on domain knowledge and data observations, we succeed to use the correlation graph between channels as the cornerstone of our model. 
The proposed instantaneous and delayed time shift tasks help us capture the correlation patterns of brain signals spatially and temporally. 
The replace discriminative task helps \model learn a unique representations for each channel to achieve accurate channel-wise prediction.
Extensive experiments of seizure detection on large-scale real-world datasets demonstrate the superior performance and clinical value of \model. 
However, there are still some limitations of our work. 
For example, negative sampling of multi-channel CPC consumes certain memory and time. 
As for the future work, we plan to collect more types of brain signals and extend \model to more downstream tasks.

\vpara{Acknowledgment.} This work is supported by NSFC (No.62176233), the National Key Research and Development Project of China (No.2018AAA0101900) and the Fundamental Research Funds for the Central Universities.

\bibliographystyle{ACM-Reference-Format}
\bibliography{reference}

\clearpage
\appendix
\section{Preliminaries} \label{app:segment}

\vpara{Brain signal data.}
For both EEG and SEEG data, there are multiple \textit{electrodes} with $\mathbf{C}$ contacts that are sampled at a fixed frequency to record the brain signals. We also call these contacts \textit{channels}. For every sampling point, each channel records the potential value of the brain region in which they are located, constituting abstract multi-channel time series data. A complete record file contains a total of $\mathbf{L}$ time points, for which we use the notation $X = \{ x_{l} \in \mathbb{R}^{\mathbf{C}} \}_{l=1}^{\mathbf{L}}$ to represent. In this paper, we use $i$ and $j$ to denote the indexes of channels, such as $x_{l} = \{ x_{l,i} \}_{i=1}^{\mathbf{C}}$. For every $x_{l, i}$, we assign a binary label $Y_{l, i} \in \{ 0, 1\}$ to it according to the start and end time of seizure signals marked by doctors. The time points are in the seizure state with positive labels ($Y_{l, i} = 1$), while zero labels ($Y_{l, i} = 0$) represent the normal data.

\vpara{Preprocessing.}
Following the existing time series works~\citep{Zhu2017preproseg, bagnall2017great, schafer2015boss} with the common preprocessing of segmentation, we use a $\mathbf{W}$-length window to divide the original data $X$ into time segments $S = \{ s_{t} \in \mathbb{R}^{\mathbf{W} \times \mathbf{C}} \}_{t=1}^{|S|}$ without overlapping. The number of segments $|S| = \lfloor \mathbf{L}/\mathbf{W} \rfloor$. The segment label is obtained from the time points of the whole segment, \emph{i.e.}, $Y^{s}_{t, i} = \max\{ Y_{t \times \mathbf{W} + 1, i}, \dots, Y_{(t + 1) \times \mathbf{W}, i} \}$.

\section{Single-channel CPC} \label{app:cpc}

Contrastive Predictive Coding (CPC), a pioneering model for self-supervised contrastive learning, sets the pretext task to predict low-level local representations by high-level global contextual information $c_{t}$. In this way, the model can avoid learning too many details of the raw data and pay more attention to the contextual semantic information. The InfoNCE loss proposed in CPC has become the basic design of the contrastive learning loss function. Formally, given a raw data sample set $X=\{x_{1}, \dots, x_{N}\}$ consisting of one positive sample from $p(x_{t+k}|c_{t})$ and $N-1$ negative samples from the noisy distribution $p(x_{t+k})$, InfoNCE will optimize:
\begin{equation}
    \mathcal{L}_{N} = -\mathbb{E}_{X}\left[ \log{\frac{f_{k}(x_{t+k}, c_{t})}{\sum_{x_{j} \in X} f_{k}(x_{j}, c_{t})}} \right]. \label{eq:single_cpc}
\end{equation}
In order to obtain the best classification probability of the positive sample with the cross entropy loss function, the optimal $f_{k}(x_{t+k}, c_{t})$ is proportional to $p(x_{t+k}|c_{t})/p(x_{t+k})$. Furthermore, the optimal loss function is also closely related to mutual information, as below:
\begin{align}
    \mathcal{L}_{N}^{\text{opt}} &= -\mathbb{E}_{X}\left[ \log{\frac{p(x_{t+k}|c_{t})/p(x_{t+k})}{p(x_{t+k}|c_{t})/p(x_{t+k}) + \sum_{x_{j} \in X_{\text{neg}}} p(x_{j}|c_{t})/p(x_{j})}} \right] \nonumber \\
    &\ge \mathbb{E}_{X}\left[ \log{\frac{p(x_{t+k})}{p(x_{t+k}|c_{t})} N} \right] ~\label{eq:inequ} \\
    &= -I(x_{t+k};c_{t}) + \log{N}.
\end{align}
Therefore, we can conclude that while minimizing the loss function $\mathcal{L}_{N}$, we are also constantly approximating the mutual information of raw data distribution $p(x_{t+k})$ and contextual semantic distribution $p(c_{t})$.
It turns out that InfoNCE is indeed a well-established loss function designed for self-supervised contrastive learning.

\section{Implementation Details of \model} \label{app:main_imp}

The non-linear encoder $g_{\text{enc}}$ used in \model is composed of three 1-D convolution layers and a one-layer LSTM model~\citep{lstm1997} is used as the autoregressive model $g_{\text{ar}}$. The model is optimized using Adam optimizer~\citep{kingma2015adam} with a learning rate of 2e-4 and weight decay of 1e-6 for the self-supervised learning stage. And for the downstream training stage, the downstream model is optimized with a learning rate of 5e-4 and weight decay of 1e-6 while the SSL model is fine-tuned with a low learning rate of 1e-6. 
For the hyperparameters of \model, we set $\theta_{1} = 0.5$ and $\theta_{2} = 0.5$. We set the maximum value of $k_{1}$ in instantaneous time shift task as 8. As Figure~\ref{pic:hyper_k2} shows, we set $K_{2} = 7$ so as to take into account the step with the most significant correlation in delayed time shift task.
Lastly, we build our model using PyTorch 1.8~\citep{paszke2019pytorch} and train it on a workstation with 4 NVIDIA GeForce RTX 3090.

\begin{figure}[ht]
  \centering
  \includegraphics[width=\linewidth]{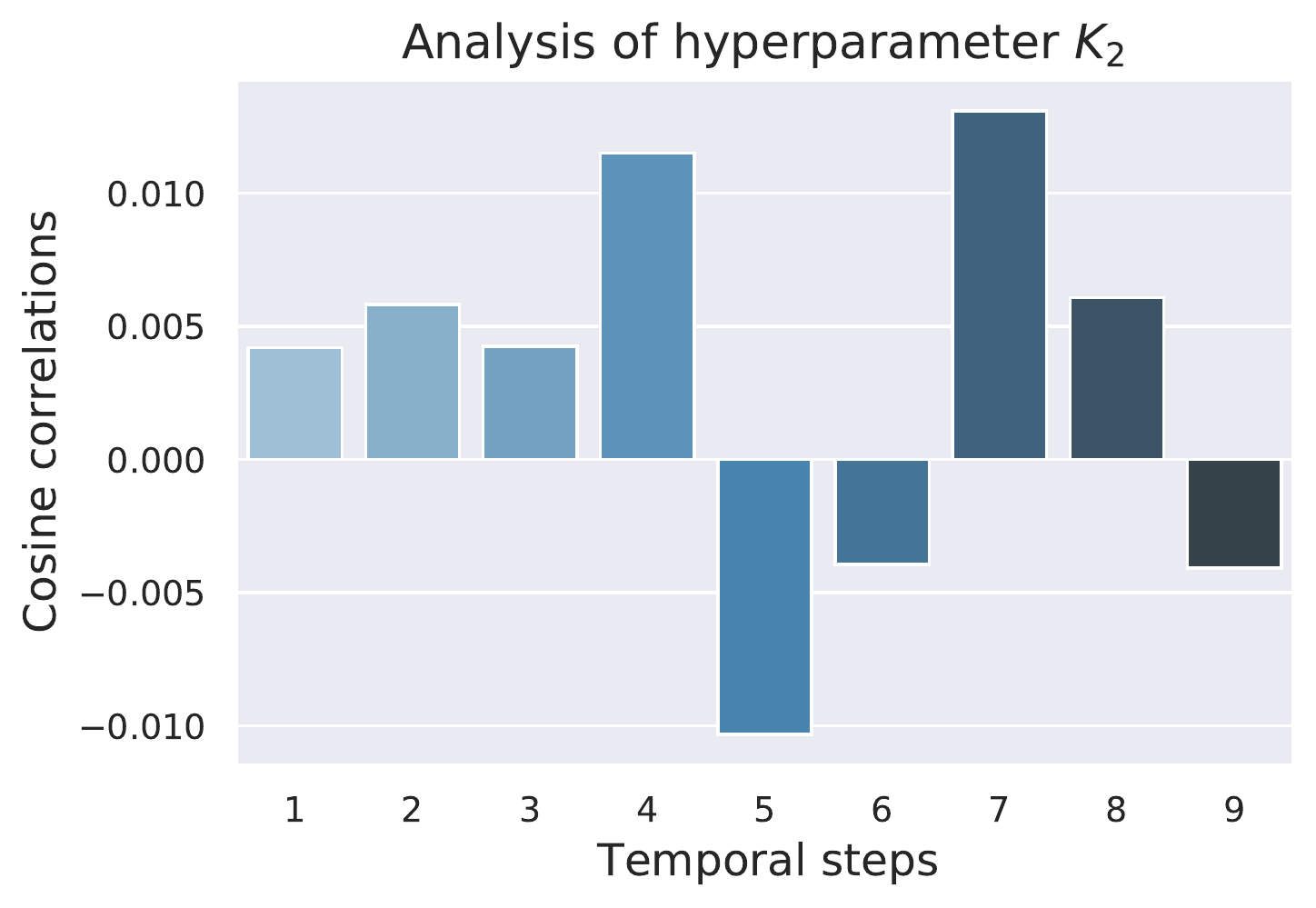}
  \caption{
  The data observation of how to choose hyperparameter $K_{2}$.
  \small
  We first average the correlations between each channel and all other channels in each time step. Then we average those of all channels in the same time step.
  }
  \label{pic:hyper_k2}
\end{figure}

For the downstream task, we first utilize an LSTM model~\citep{lstm1997} to encode the segment representations of each channel in chronological order independently. One-layer self-attention~\citep{Ashish2017attention} is then adopted to all channels within the same time step. Finally, a two-layer MLP classifier is used to predict whether seizure is occurring in the time segments. All baselines share the same downstream model in our experiments.

\section{Details of Ablation Study} \label{app:ablation}

\textbf{Replace the method to aggregate channel information.}
We design two strategies to combine multi-channel CNN or MLP into CPC respectively to learn representations for each channel.
\begin{itemize}[leftmargin=*]
    \item Directly use 1-Dimension CNN to encode the whole time series data and the number of channels during the process is $\mathbf{C} \rightarrow 256 \rightarrow 256 \rightarrow \mathbf{C} \times 256$, and split the output into $\mathbf{C}$ representations, each of which is a 256-dimensional representation. Then an LSTM is implemented to it. Then we execute the self-supervised task and the downstream task of CPC based on the representations for each channel as \model does, this variant is denoted as {CPC\small{-Conv}}.

    \item We use the contextual representations of all $n$ channels as input to an MLP in a fixed order, but we set the representation of the target channel to $0$ tensor when we aggregate them. By using the output of MLP as the aggregated representation of other channels, we perform subsequent experiments following exactly the same steps as \model. We name this variant as {CPC\small{-MLP}}.
\end{itemize}

\section{Hyperparameter Analysis} \label{app:hyper}

\begin{figure}[ht]
  \centering
    \subfigure[Weights search for $\mathcal{L}_2$.]{
    \begin{minipage}[ht]{\linewidth}
    \centering
    \includegraphics[width=7cm]{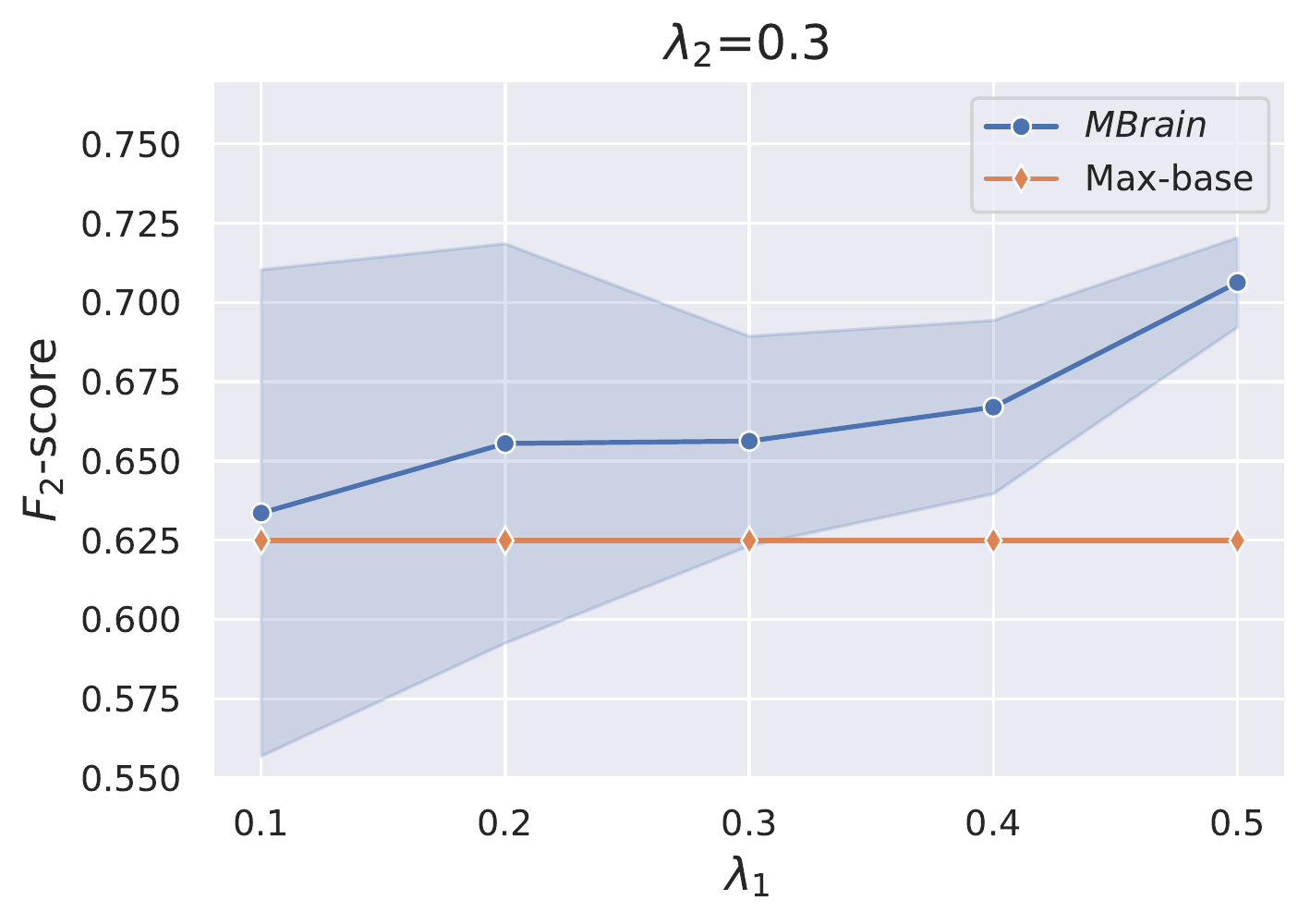}
    \end{minipage}%
    \label{subpic:tuning_lambda1}}%
    \\
    \subfigure[Weights search for $\mathcal{L}_3$.]{
    \begin{minipage}[ht]{\linewidth}
    \centering
    \includegraphics[width=7cm]{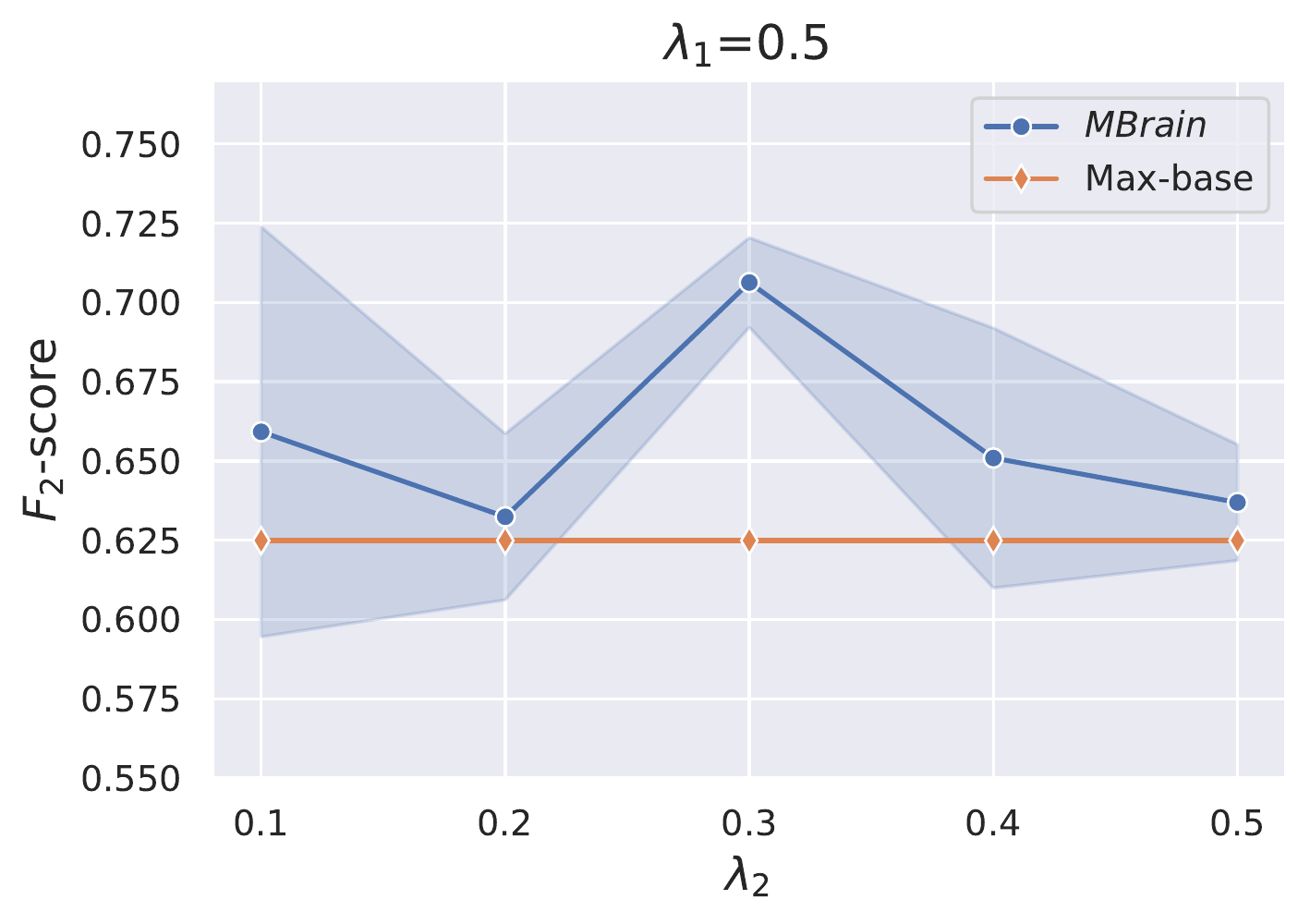}
    \end{minipage}%
    \label{subpic:tuning_lambda2}}%
  \caption{
  Sensitivity analysis on loss weights.
  }
  \label{pic:tuning_lambda}
\end{figure}

\vpara{Sensitivity analysis on loss weights.}
Our loss function is defined as: $\mathcal{L} = (1-\lambda_1-\lambda_2)\mathcal{L}_1 + \lambda_1\mathcal{L}_2 + \lambda_2\mathcal{L}_3$, where $\mathcal{L}_1$, $\mathcal{L}_2$ and $\mathcal{L}_3$ are the loss of instantaneous time shift prediction task, delayed time shift prediction task and replace discriminative task respectively, and $\lambda_1$ and $\lambda_2$ are hyperparameters to balance the three pre-training tasks. We search both of the weights of $\lambda_1$ and $\lambda_2$ in the set \{0.1, 0.2, 0.3, 0.4, 0.5\} and report the tuning results with $F_2$-score for seizure detection task on subject-A from SEEG dataset.
In~\ref{subpic:tuning_lambda1} and~\ref{subpic:tuning_lambda2}, we can see that $\lambda_1$ = 0.5 and $\lambda_2$ = 0.3 lead to the optimal performance. In addition, \model consistently performs better than the best baseline.

\begin{figure}[ht]
  \centering
  \includegraphics[width=\linewidth]{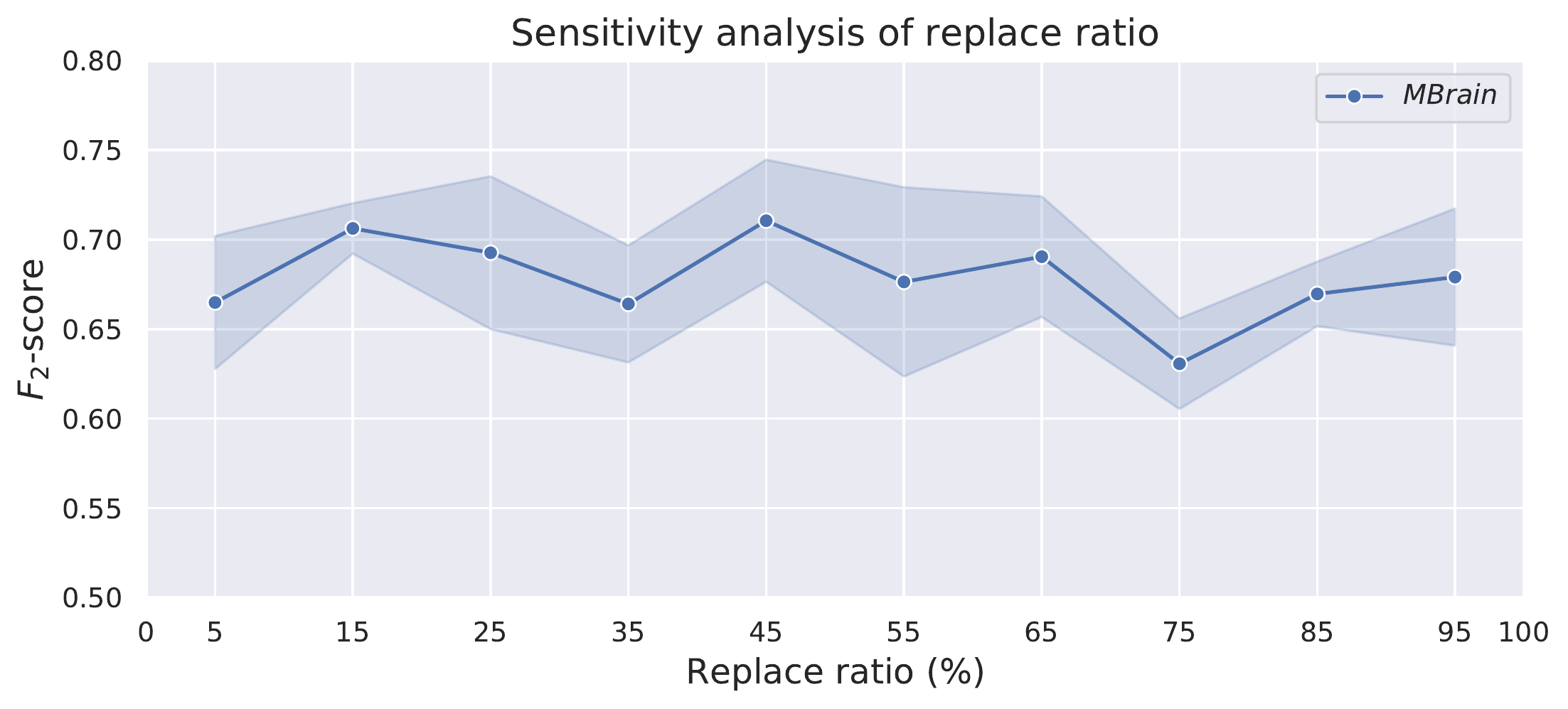}
  \caption{
  Sensitivity analysis on replace ratio $r\%$.
  }
  \label{pic:replace_ratio}
\end{figure}

\vpara{Sensitivity analysis on replace ratio.}
We perform sensitivity analysis on replace ratio $r\%$ from replace discriminative task. We search the replace ratio from $5\%$ to $95\%$ and report the tuning results with $F_2$-score for seizure detection task on subject-A from SEEG dataset. As Figure~\ref{pic:replace_ratio} shows, when the replace ratio is set as $45\%$, \model has the best performance of 71.06$\pm$3.41. While \model gets the smallest standard deviation and the second best performance of 70.63$\pm$1.41 when the replace ratio is set as $15\%$.

\section{Emotion Recognition Task} \label{app:emotion}

To measure the performance of our model on different downstream tasks, we use the SJTU Emotion EEG Dataset (SEED)~\citep{duan2013differential} to test the model's performance in the emotion recognition task.
In SEED, fifteen Chinese film clips (positive, neutral and negative emotions) were chosen from the pool of materials as stimuli used in the experiments. The duration of each film clip is approximately 4 minutes. We divide each EEG segment into 24-second segments without overlapping. For experimental efficiency, we downsample the segments to half the original frequency for each 24-second EEG segment. We randomly split the SEED dataset by subjects into train set, valid set and test set at a ratio of 3:1:1. We sample 3500 and 2000 EEG clips from the training patients for SSL and downstream task. We then sample 500 clips as validation set. Finally, we use all the data from the testing patients to evaluate models.

\begin{table}[ht]
    \caption{The performance of models on SEED dataset.}
    \label{tab:SEED}
    \centering
    \begin{tabular}{lcc}
        \toprule
        \textbf{Models} & Acc. & AUROC \\
        \midrule
        MiniRocket & 49.80\small{$\pm$0.60} & 75.28\small{$\pm$0.17} \\ 
        \midrule
        CPC & 48.23\small{$\pm$4.36} & 73.48\small{$\pm$1.51} \\
        SimCLR & 44.84\small{$\pm$5.82} & 63.05\small{$\pm$6.52} \\
        T-Loss & 47.90\small{$\pm$3.99} & 68.56\small{$\pm$5.96} \\
        TST & 35.13\small{$\pm$0.34} & 53.49\small{$\pm$1.26} \\
        GTS & 39.85\small{$\pm$0.34} & 60.18\small{$\pm$1.30} \\
        TS-TCC & 40.10\small{$\pm$5.50} & 66.38\small{$\pm$3.39} \\
        TS2Vec & 48.75\small{$\pm$2.74} & 71.60\small{$\pm$2.16} \\
        \midrule
        \model & \textbf{52.44\small{$\pm$1.21}} & \textbf{75.52\small{$\pm$1.27}} \\
        \bottomrule
    \end{tabular}
\end{table}

Table~\ref{tab:SEED} shows the results of \model and all baseline models on the emotion recognition task on SEED dataset. Since this is a 3-class classification task with balanced samples for each class, we only report the two metrics of Accuracy (Acc.) and AUROC. 
As can be seen from the table, \model improves the Accuracy by 5.30\% and the AUROC by 0.32\% on SEED dataset compared to the best result of baseline methods, demonstrating that \model still has a good performance on other downstream tasks of brain signals. 
It is worth noting that, like the TUSZ dataset, the SEED dataset only has coarse labels for each EEG clip, but our proposed \model aims to learn fine-grained representations. Therefore, the performance improvement on the SEED dataset is not as obvious as that on the SEEG dataset, but \model still outperforms all baseline models.

\section{Implementation Details of Baselines}\label{app:baseline}

\begin{itemize}[leftmargin=*]
    \item \textbf{MiniRocket}~\citep{dempster2021minirocket}: Rocket~\citep{dempster2020rocket} is a state-of-the-art supervised time series classification method based on evaluations on public benchmarks~\citep{bagnall2017great, tan2020monash}, involves training a linear classifier on top of features extracted by a flat collection of numerous and various random convolutional kernels. MiniRocket is a variant of Rocket which improves processing time, while offering essentially the same accuracy.
    We use the open source code from \href{https://github.com/angus924/minirocket}{https://github.com/angus924/minirocket}. For each subject, we use the features obtained through MiniRocket to train an independent logistic regression classifier for each channel and test it on the test set of that channel.

    \item \textbf{CPC}~\citep{van2018representation}: This is a self-supervised learning method based on a contrastive loss InfoNCE. The pretext task of CPC is set to predict future local low-level representations obtained from multi-layer CNNs by contextual high-level representations obtained from an autoregressive model. This is the backbone model in this paper. We use the open source code of the corrected version from \href{https://github.com/facebookresearch/CPC_audio}{https://github.com/facebookresearch/CPC\_audio}.
    
    \item \textbf{SimCLR}~\citep{chen2020simple}: This is a simple yet effective framework for contrastive learning of visual representations and we use time-series speciﬁc augmentations to adapt it to our application. We implemented SimCLR on time series data by ourselves. We use the same encoder architecture and parameter configuration as TS-TCC. In the meantime, we also follow TS-TCC and use scaling (sigma=1.1) as the data augmentation way.
    
    \item \textbf{Triplet-Loss (T-Loss)}~\citep{franceschi2019unsupervised}: The approach employs time-based negative sampling and a triplet loss to learn representations for time series segments. We use the default model architecture from the source code provided by the author (\href{https://github.com/White-Link/UnsupervisedScalableRepresentationLearningTimeSeries}{https://github.com/White-Link/UnsupervisedScalableRepresentationLearningTimeSeries}). For the sampling method of negative samples, we use the data of the previous batch as the candidate set of negative samples of the current batch data (the negative sample candidate set for the first batch is itself). Since the dataloader is shuffled at the end of each epoch, there is no need to worry about the case where the set of sampled negative samples does not change.
    
    \item \textbf{Time Series Transformer (TST)}~\citep{zerveas2021transformer}: This is a unsupervised representation learning framework for multivariate time series by training a transformer model to extract dense vector representations of time series through an input denoising objective. We use the default model architecture from the source code provided by the author (\href{https://github.com/gzerveas/mvts_transformer}{https://github.com/gzerveas/mvts\_transformer}).
    
    \item \textbf{GTS}~\citep{shang2021discrete}: This is a time series forecasting model that learns a graph structure among multiple time series and forecasts them simultaneously with a GNN. In view of this, this model can learn useful representations from unlabeled time series data. We use the default model architecture from the source code provided by the author (\href{https://github.com/chaoshangcs/GTS}{https://github.com/chaoshangcs/GTS}). In the pre-training stage, we divide each time series segment into 10 parts on average, and learn a time series forecasting model that predicts the next 2 steps based on the previous 8 steps. In the downstream task stage, we use the representation after step 10 as the representation of the time series segment for the seizure detection task.
    
    \item \textbf{TS-TCC}~\citep{eldele2021time}: This is an unsupervised time-series representation learning framework, applying a temporal contrasting module and a contextual contrasting module to learn robust and discriminative representations. We use the default model architecture from the open source code provided by the author (\href{https://github.com/emadeldeen24/TS-TCC}{https://github.com/emadeldeen24/TS-TCC}).
    
    \item \textbf{TS2Vec}~\citep{yue2021ts2vec}: This is a universal representation learning framework for time series, that applies hierarchical contrasting to learn scale-invariant representations within augmented context views. We use the default model architecture from the source code provided by the author (\href{https://github.com/yuezhihan/ts2vec}{https://github.com/yuezhihan/ts2vec}).
\end{itemize}

\end{document}